\documentclass[journal]{IEEEtran}
\usepackage{mathrsfs}
\usepackage{amsmath}
\usepackage{mathrsfs}
\usepackage{stmaryrd}
\usepackage{amsfonts}
\usepackage{algorithm}
\usepackage{algorithmic}
\usepackage{subfigure}
\usepackage{cite}
\usepackage{blindtext}
\usepackage{amssymb}
\usepackage{esint}
\usepackage{graphicx,xcolor,subfig}
\usepackage{epstopdf}
\usepackage[section] {placeins}
\usepackage{morefloats}
\graphicspath{{./pics/}}
\usepackage{eepic,epic,epsfig,epstopdf}
\usepackage{url}

\newtheorem{theorem}{Theorem}
\newtheorem{remark}{Remark}[section]
\newtheorem{assumption}{Assumption}[section]
\newtheorem{lemma}[theorem]{Lemma}
\newtheorem{proposition}[theorem]{Proposition}
\newtheorem{corollary}[theorem]{Corollary}

\def\calA{{\mathcal{A}}}
\def\calI{{\mathcal{I}}}

\ifCLASSINFOpdf
\else
\fi

\begin{document}

\title{Group Sparse Recovery via the $\ell^0(\ell^2)$ Penalty: Theory and Algorithm}
\author{Yuling Jiao,\thanks{School of Statistics and Mathematics, Zhongnan University of Economics and Law, Wuhan, 430063, P.R. China. (yulingjiaomath@whu.edu.cn)}\quad
Bangti Jin,\thanks{Department of Computer Science, University College London, Gower Street, London WC1E 6BT, UK. (bangti.jin@gmail.com, b.jin@ucl.ac.uk)}
\quad\and Xiliang Lu\thanks{Corresponding author. School of Mathematics and Statistics and
Hubei Key Laboratory of Computational Science, Wuhan University, Wuhan 430072, P.R. China. (xllv.math@whu.edu.cn)}
}

\markboth{IEEE TRANSECTION ON SIGNAL PROCESSING,Vol.~~, No.~~, ~~,2015}%
{Shell \MakeLowercase{\textit{et al.}}: Bare Demo of IEEEtran.cls for Journals}
\maketitle

\begin{abstract}
In this work we propose and analyze a novel approach for group sparse recovery. It is based on regularized least squares
with an $\ell^0(\ell^2)$ penalty, which penalizes the number of nonzero groups. One distinct
feature of the approach is that it has the built-in decorrelation mechanism within each group,
and thus can handle challenging strong inner-group correlation. We provide a complete analysis of
the regularized model, e.g., existence of a global minimizer, invariance property, support
recovery, and properties of block coordinatewise minimizers.  Further, the regularized
problem admits an efficient primal dual active set algorithm
with a provable finite-step global convergence. At each iteration, it involves solving a least-squares
problem on the active set only, and exhibits a fast local convergence, which makes the method extremely efficient for recovering group sparse
signals. Extensive numerical experiments are presented to illustrate salient features of the
model and the efficiency and accuracy of the algorithm. A comparative study indicates
its competitiveness with existing approaches.
\end{abstract}
\begin{IEEEkeywords}
group sparsity, block sparsity, blockwise mutual incoherence, global minimizer,
block coordinatewise minimizer, primal dual active set algorithm, $\ell^0(\ell^2)$ penalty
\end{IEEEkeywords}

%
\IEEEpeerreviewmaketitle

\section{Introduction}\label{sec:intro}
\IEEEPARstart{S}{parse} recovery has received much attention in many areas,
e.g., signal processing, statistics, and machine learning
recently. The key assumption is that the
data $y\in\mathbb{R}^n$ is generated by a linear combination of a few
atoms of a given dictionary $\Psi\in\mathbb{R}^{n\times p}$, $p\gg n$, where
each column represents an atom. In the
presence of noise $\eta\in\mathbb{R}^n$ (with a noise level $\epsilon=\|\eta\|$),
it is formulated as
\begin{equation}\label{model}
  y = \Psi x^{\dag} + \eta,
\end{equation}
where the vector $x^\dag \in \mathbb{R}^{p}$ denotes the
signal to be recovered.

The most natural formulation of the problem of finding the sparsest solution is the following $\ell^0$ optimization
\begin{equation}\label{eqn:l0reg}
   \min_{x \in \mathbb{R}^{p}} \tfrac{1}{2}\|\Psi x-y\|^2 + \lambda \|x\|_{\ell^0},
\end{equation}
where $\|\cdot\|$ denotes the Euclidean norm of a vector, $\|\cdot\|_{\ell^0}$ denotes the number of
nonzero entries, and $\lambda>0$ is a regularization parameter. Due to discontinuity of the $\ell^0$ penalty,
it is challenging to find a global minimizer of problem \eqref{eqn:l0reg}. In practice, lasso / basis pursuit
\cite{Tibshirani:1996,Chen:1998}, which replaces the $\ell^0$ penalty with its convex
relaxation, the $\ell^1$ penalty, has been very popular. Many deep results
on the equivalence between the $\ell^0$ and $\ell^1$ problems and error estimates have been obtained
\cite{CandesTao:2005,CandesRombergTao:2006}, based on
the concepts mutual coherence (MC) and restricted isometry
property (RIP).

\subsection{Group sparse recovery}

In practice, in addition to sparsity, signals may exhibit additional structure, e.g.,
nonzero coefficients occur in clusters/groups, which are commonly known as block- / group-sparsity.
In electroencephalogram (EEG), each group encodes the information about the
direction and strength of the dipoles of each discrete voxel representing the dipole
approximation \cite{OuHamalainenGolland:2009}. Other applications include multi-task
learning \cite{ArgyriouEvgeniouPontil:2008}, wavelet image analysis \cite{Shapiro:1993,
AntoniadisFan:2001}, gene analysis \cite{HeYu:2010,MaSongHuang:2007} and multichannel
image analysis \cite{MishaliEldar:2009,MishaliEldar:2010}, to name a few.
The multiple measurement vector problem is also one special case \cite{ChenHuo:2006}.
In these applications, the focus is to recover all contributing groups, instead of one
entry from each group. The group structure is an important piece
of \textit{a priori} knowledge about the problem, and should be properly accounted for in the recovery
method in order to improve interpretability and accuracy of the recovered signal.

There have been many important developments of
group sparse recovery. One popular approach is group lasso, extending
lasso using an $\ell^1(\ell^2)$-penalty
\cite{Bakin:1999,MalioutovCetinWillsky:2005,YuanLin:2006,HuangBrehenyMa:2012}. A number of
theoretical studies have shown many desirable properties of group lasso, and its
advantages over lasso for recovering group sparse signals
\cite{HuangZhang:2010,BaraniukCevher:2010,Eldar:2010,LouniciPontil:2011,Bajwa:2015,ErenVidyasagar:2015}
under suitable MC or RIP type conditions.
To remedy the drawbacks of group lasso, e.g., biasedness and lack of the oracle
property \cite{FanLi:2001,ZhangZhang:2012}, nonconvex penalties
have been extended to the group case, e.g., bridge, smoothly clipped absolute deviation (SCAD),
and minmax concavity penalty (MCP) \cite{WangLiHuang:2008,HuangMaZhang:2009,HuangBrehenyMa:2012}.
A number of efficient algorithms \cite{YuanLin:2006,MeierVandegeerPeter:2008,VandDenBerg:2008,TsengYun:2009,ChenLinKim:2011,
She:2012,QinScheinberGoldfarb:2013,BrehenyHuang:2015} have been proposed for convex and nonconvex
group sparse recovery models. Like in the sparse case, several group greedy methods have also been
developed and analyzed in depth \cite{Eldar:2010,BenElda:2011,GaneshZhouMa:2009}.

However, in these interesting works, the submatrices of $\Psi$ are assumed to be well conditioned
in order to get estimation errors. While this assumption is reasonable in some applications, it
excludes the practically important case of strong correlation within groups. For example, in microarray
gene analysis, it was observed that genes in the same pathway produce highly correlated values
\cite{SegalDahlquist:2004}; in genome-wide association studies, SNPs are highly correlated or even
linearly dependent within segments of the DNA sequence \cite{Balding:2006}; in functional neuroimaging,
identifying the brain regions involved in the cognitive processing of an external stimuli is formalized
as identifying the non-zero coefficients of a linear model predicting the external stimuli from the neuroimaging
data, where strong correlation occurs  between  neighboring voxels \cite{Tom:2007}; just to name a few.

In the presence of strong inner-group correlation, an inadvertent application
of standard sparse recovery techniques is unsuitable. Numerically, one often can only
recover one predictor within each contributing group, which is undesirable when seeking the whole
group \cite{ZouHastie:2005}. Theoretically, the correlation leads bad RIP or MC conditions,
and thus many sparse recovery techniques may perform poorly.

\subsection{The $\ell^0(\ell^2)$ approach and our contributions}

In this work, we shall develop and analyze a nonconvex model and algorithm for recovering
group-sparse signals with potentially strong inner-group
correlation. Our approach is based on the following $\ell^0(\ell^2)$ optimization
\begin{equation}\label{eqn:groupl0}
   \min_{x \in \mathbb{R}^{p}} \left\{J_\lambda(x) =  \tfrac{1}{2}\|\Psi x-y\|^2 + \lambda \|x\|_{\ell^0(\ell^2)}\right\},
\end{equation}
where the $\ell^0(\ell^2)$ penalty $\|\cdot\|_{\ell^0(\ell^2)}$ (with respect to a given partition
$\{G_i\}_{i=1}^N$) is defined below in \eqref{eqn:l0l2}, and the regularization parameter $\lambda>0$
controls the group sparsity level of the solution.
The $\ell^0 (\ell^2)$ penalty is to penalize the number of nonzero
groups.
To the best of our knowledge, this model has not been systematically studied in the literature, even
though the $\ell^0(\ell^2)$ penalty was used in several prior works; see
Section \ref{sec:existing} below.
We shall provide both theoretical analysis and efficient solver for the model.

The model \eqref{eqn:groupl0} has several distinct features.
The regularized solution is invariant under full rank column transformation,
and does not depend on the specific parametrization within the groups. Thus, it allows strong
inner-group correlation and merits a built-in decorrelation effect, and admits theoretical results
under very weak conditions. Further, both global minimizer and block coordinatewise minimizer have
desirable properties, e.g., support recovery and oracle property.

The main contributions of this work are three-folded.
First, we establish fundamental properties of the model \eqref{eqn:groupl0},
e.g., existence of a global minimizer, local optimality, necessary optimality condition, and
transformation invariance, which theoretically substantiates \eqref{eqn:groupl0}. For example, the
invariance implies that it can be equivalently transformed into
a problem with orthonormal columns within each group, and thus it is independent of
the conditioning of inner-group columns, which contrasts sharply with most existing group
sparse recovery models. Second, we develop an efficient algorithm for solving the
model \eqref{eqn:groupl0}, which is of primal dual active set (PDAS) type. It is based on
a careful analysis of the necessary optimality system, and represents a nontrivial
extension of the PDAS algorithm for the $\ell^1$ and $\ell^0$ penalties \cite{FanJiaoLu:2014,JiaoJinLu:2014}.
It is very efficient when coupled with a continuation strategy, due to its Newton nature \cite{FanJiaoLu:2014}.
Numerically, each inner iteration involves only solving a least-squares problem on the active
set. The whole algorithm converges globally in finite steps to
the oracle solution. Third, we present extensive numerical experiments
to illustrate the features of our approach, and to show its competitiveness with
start-of-art group sparse recovery methods, including group lasso and greedy methods.

\subsection{Connections with existing works and organization}\label{sec:existing}
The proposed model \eqref{eqn:groupl0} is closely related to the following constrained
nonconvex optimization
\begin{equation}\label{eqn:l0lq}
  \min \|x\|_{\ell^0(\ell^q)}\quad \mbox{subject to } y = \Psi x, \tag{$P_q$}
\end{equation}
in the absence of noise $\eta$. This model was studied in
\cite{EldarMishali:2009,Eldar:2010,GaneshZhouMa:2009,ElhamifarVidal:2012}.
 In the case of $q=2$, Eldar and Mishali
\cite{EldarMishali:2009} discussed unique group sparse recovery,
and Eldar et al \cite{Eldar:2010} developed an orthogonal matching pursuit algorithm
for recovering group sparse signals and established recovery condition in terms of
block coherence. See also \cite{GaneshZhouMa:2009} for related results for subspace
signal separation.  Elhamifar and Vidal \cite{ElhamifarVidal:2012}
derived the necessary and sufficient conditions for the equivalence of problem
\eqref{eqn:l0lq} with a convex $\ell^1(\ell^q)$ relaxation, and gave sufficient
conditions using the concept cumulative subspace coherence. Further,
under even weaker conditions, they extended these results to the $\Psi$-weighted formulation
\begin{equation}\label{eqn:l0lq2}
  \min \sum_{i=1}^N\|\Psi_{G_i} x_{G_i}\|_{\ell^q}^0 \quad \mbox{subject to } y =\Psi x, \tag{$P_q^\prime$}
\end{equation}
which is especially suitable for redundant dictionaries. The models \eqref{eqn:l0lq} and
\eqref{eqn:l0lq2} are equivalent, if the columns within each group are of full column rank.
Our approach \eqref{eqn:groupl0} can be viewed as a natural extension of \eqref{eqn:l0lq}
with $q=2$ to the case of noisy data using a Lagrangian formulation, which, due to the
nonconvexity of the $\ell^0(\ell^2)$ penalty, is generally not equivalent to the constrained
formulation. In this work, we provide many new insights into analytical properties and
algorithm developments for the model \eqref{eqn:groupl0}, which have not been discussed in
these prior works. Surprisingly, we shall show that the model \eqref{eqn:groupl0} has
built-in decorrelation effect for redundant dictionaries, similar to the model \eqref{eqn:l0lq2}.

The rest of the paper is organized as follows. In Section \ref{sec:prelim}, we describe the problem setting,
and derive useful estimates. In Section \ref{sec:cwm}, we provide analytical properties, e.g., the
existence of a global minimizer, invariance property, and optimality condition. In Section \ref{sec:pdasc},
we develop an efficient group primal dual active set with continuation algorithm, and analyze its
convergence and computational complexity. Finally, in Section \ref{sec:numer}, several numerical examples are provided to
illustrate the mathematical theory and the efficiency of the algorithm. All the technical proofs are given in the appendices.

\section{Preliminaries}\label{sec:prelim}
In this section, we describe the problem setting, and derive useful estimates.

\subsection{Problem setting and notations}
Throughout, we assume that the sensing matrix $\Psi\in \mathbb{R}^{n\times p}$ with $n\ll p$ has normalized
columns $\|\psi_i\| =1$ for $i=1,...,p$, and the index set ${S} = \{1,...,p\}$ is divided
into $N$ non-overlapping groups $\{G_i\}_{i=1}^N$ such that $1 \leq s_i = |G_i| \leq s$ and $\sum_{i=1}^N|G_i| = p$.
For any index set $B\subseteq S$, we denote by $x_B$ (respectively $\Psi_B$) the subvector of $x$ (respectively the submatrix of $\Psi$)
which consists of the entries (respectively columns) whose indices are listed in $B$.
All submatrices $\Psi_{G_i}$, $i=1,2,\ldots,N$, are assumed to have full column rank.  The true signal $x^\dag$ is
assumed to be group sparse with respect to the partition $\{G_i\}_{i=1}^N$, i.e., $x^\dag= (x^\dag_{G_1},...,x^\dag_{G_N})$,
with $T$ nonzero groups. Accordingly, the group index set $\{1,\ldots,N\}$ is divided into the active set $\calA^\dag$ and
inactive set $\calI^\dag$ by
\begin{equation}\label{eqn:active}
\mathcal{A}^\dag = \{i: \|x^\dag_{G_i}\| \neq 0\} \quad \mbox{and}\quad \mathcal{I}^\dag = (\mathcal{A}^\dag)^c.
\end{equation}
The data vector $y$ in \eqref{model}, possibly contaminated by noise, can be recast as
$
  y = \Psi x^\dag + \eta = \sum_{i\in \mathcal{A}^\dag}\Psi_{G_i} x^\dag_{G_i} + \eta.
$
Given the true active set  $\mathcal{A}^\dag$ (as if it were provided by an oracle), we define the oracle
solution $x^o$ by the least squares solution on $\calA^\dag$ to \eqref{model}, i.e.,
\begin{equation}\label{eqn:oracle}
x^o = \mathop\textrm{argmin}_{\mathrm{supp}(x) \subseteq \cup_{i\in\calA^\dag}G_i } \|\Psi x - y\|^2.
\end{equation}
The oracle solution $x^o$ is uniquely defined provided that $\Psi_{\cup_{i\in \mathcal{A}^\dag}G_i}$ has
full column rank. It is the best approximation for problem \eqref{model}, and will be
used as the benchmark.

For any vector $x\in \mathbb{R}^p$, we define an $\ell^r(\ell^q)$-penalty (with respect to the partition
$\{G_i\}_{i=1}^N$) for $r\geq0$ and $q>0$ by
\begin{equation}\label{eqn:l0l2}
\|x\|_{\ell^r(\ell^q)} = \left\{\begin{array}{ll}(\sum_{i=1}^N \|x_{G_i}\|_{\ell^q}^r)^{1/r}, & r >0, \\
\sharp\{i:  \|x_{G_i}\|_{\ell^q} \neq 0\}, & r=0, \\
\max_{i}\{\|x_{G_i}\|_{\ell^q}\}, & r=\infty.
\end{array}\right.
\end{equation}
When $r=q>0$, the $\ell^r(\ell^q)$ penalty reduces to the usual $\ell^r$ penalty. The choice $r=0$ (or
$r=\infty$) and $q=2$ is frequently used below. Further, we shall abuse the notation $\|\cdot
\|_{\ell^r(\ell^q)}$ for any vector that is only defined on some sub-groups (equivalently zero extension).

For any $r,q\geq1$, the $\ell^r(\ell^q)$ penalty defines a proper norm, and was
studied in \cite{Kowalski:2009}. For any $r,q>0$,
the $\ell^r(\ell^q)$ penalty is continuous. The $\ell^0(\ell^2)$ penalty,
which is of major interest in this work, is discontinuous, but still lower semi-continuous. 
\begin{proposition}\label{prop:lsc}
The $\ell^0(\ell^2)$ penalty is lower semicontinuous.
\end{proposition}
\begin{proof}
Let $\{x^n\}\subset\mathbb{R}^p$ be a convergent sequence to some $x^*\in\mathbb{R}^p$. By the continuity of the $\ell^2$ norm,
$\|x^n_{G_i}\|$ converges to $\|x^*_{G_i}\|$, for $i=1,\ldots,N$. Now the assertion
follows from
$ \|x^*_{G_i}\|_{\ell^0}\leq \liminf \|x^n_{G_i}\|_{\ell^0}$ \cite[Lemma 2.2]{ItoKunisch:2014}.
\end{proof}

Now we derive the hard-thresholding operator $x^* \in H_\lambda (g)$
for one single group for an $s$-dimensional vector $g\in\mathbb{R}^s$ as
\begin{equation*}
  x^* \in \arg\min_{x\in\mathbb{R}^s} \tfrac{1}{2}\|x-g\|^2 + \lambda \|x\|_{\ell^0(\ell^2)},
\end{equation*}
where the $\|\cdot\|_{\ell^0(\ell^2)}$ penalty is given by $\|x\|_{\ell^0(\ell^2)} = 1$ if $x\neq 0$, and
$\|x\|_{\ell^0(\ell^2)} = 0$ otherwise. Then it can be verified directly
\begin{equation*}
  x^* =\left\{\begin{array}{ll}
    \ g,\ & \mbox{ if } \|g\|>\sqrt{2\lambda},\\
   \ 0,\ & \mbox{ if } \|g\|<\sqrt{2\lambda},\\
   \ 0\mbox{ or } g,\ & \mbox{ if } \|g\|=\sqrt{2\lambda}.
  \end{array}\right.
\end{equation*}
For a vector $x\in\mathbb{R}^p$, the hard thresholding operator $H_\lambda$ (with respect to the
partition $\{G_i\}_{i=1}^N$) is defined groupwise. For $s=1$, it recovers the usual hard
thresholding operator, and hence it is called a group hard thresholding operator.

\subsection{Blockwise mutual coherence}
We shall analyze the model \eqref{eqn:groupl0}
using the concept \textit{blockwise mutual coherence} (BMC). We first introduce some notation:
\begin{equation}\label{eqn:notation}
\bar\Psi_{G_i} =(\Psi_{G_i}^t\Psi_{G_i})^\frac{1}{2}\quad\mbox{and}\quad D_{i,j} = \bar\Psi_{G_i}^{-1} \Psi_{G_i}^t\Psi_{G_j} \bar{\Psi}_{G_j}^{-1}.
\end{equation}
Since $\Psi_{G_i}$ has full column rank, $\bar\Psi_{G_i}$ is symmetric
positive definite and invertible.

The main tool in our analysis is the BMC $\mu$ of the matrix $\Psi$ with respect to the
partition $\{G_i\}_{i=1}^N$, which is defined by
\begin{equation}\label{equ:mip}
\mu = \max_{i\neq j}\mu_{i,j}, \;\; \textrm{where } \mu_{i,j}= \sup_{\substack{u\in \mathcal{N}_i\backslash\{0\}\\
 v\in\mathcal{N}_j\backslash\{0\}}} \frac{\langle u,v\rangle}{\|u\|\|v\|},
\end{equation}
where $\mathcal{N}_i$ is the subspace spanned by the columns of $\Psi_{G_i}$, i.e., $\mathcal{N}_i =
\textrm{span} \{\psi_l, l\in G_i\}\subseteq \mathbb{R}^n$. The quantity $\mu_{i,j}$ is the cosine
of the minimum angle between two subspaces $\mathcal{N}_i$  and $\mathcal{N}_j$. Thus the BMC $\mu$
generalizes the concept mutual coherence (MC) $\nu$, which is defined by $\nu=\max_{i\neq j}|\langle
\psi_i,\psi_j\rangle|$ \cite{DonohoHuo:2001}, and is widely used in the analysis of sparse recovery
algorithms \cite{TroppGilbert:2007,CaiWang:2011,JiaoJinLu:2014}. The concept BMC was already
introduced in \cite{GaneshZhouMa:2009} for separating subspace signals, and \cite{ElhamifarVidal:2012}
for analyzing convex block sparse recovery. In linear algebra, one often uses principal angles to
quantify the angles between two subspaces \cite{BjorckGloub:1973}, i.e., given $U, V\subseteq
\mathbb{R}^{n}$, the principal angles $\theta_l$  for $l = 1,2,...,\min(\mathrm{dim}U, \mathrm{dim}V)$
are defined recursively by
\begin{equation*}
  \cos(\theta_l) = \max_{\substack{u\in U,\|u\|=1,\ u \perp \mathrm{span}\{u_i\}_{i=1}^{l-1} \\ v\in V,\|v\|=1,\
  v \perp \mathrm{span}\{v_j\}_{j=1}^{l-1}}} \langle u,v\rangle.
\end{equation*}

By the definition of principal angles, $\mu_{i,j} = \cos(\theta_1)$ for $(U,V) = (\mathcal{N}_i, \mathcal{N}_j)$;
see Lemma \ref{equdef} below and \cite[pp. 603--604]{BjorckGloub:1973} for the proof. Principal angles (and hence
BMC) can be computed efficiently by QR and SVD \cite{BjorckGloub:1973}, unlike
RIP or its variants \cite{Bandeira:2013}.

\begin{lemma}\label{equdef}
Let $U_i\in\mathbb{R}^{n\times s_i}$ and $V_j\in\mathbb{R}^{n\times s_j}$  be two matrices whose columns are orthonormal basis of $\mathcal{N}_i$ and
$\mathcal{N}_j$, respectively, and $\{\theta_l\}_{l=1}^{\min(s_i,s_j)}$ be the principal angles between $\mathcal{N}_i$
and $\mathcal{N}_j$. Then, $\mu_{i,j} = \cos(\theta_1) = \sigma_{\max}(U_i^{t}V_j)$.
\end{lemma}

The next result shows that the BMC $\mu$ can be bounded from above by the MC $\nu$; see Appendix \ref{app:bmic} for the proof.
Hence, the BMC is sharper than a direct extension of the MC, since the BMC does
not depend on the inner-group correlation.
\begin{proposition}\label{prop:bmic}
Let the MC $\nu$ of $\Psi$ satisfy $(s-1) \nu < 1$. Then for the BMC $\mu$ of $\Psi$, there holds
$  \mu \leq \frac{\nu s}{1 - \nu(s-1)}.$
\end{proposition}

Below we always assume  the following condition.
\begin{assumption}\label{assump:mu}
The BMC $\mu$ of $\Psi$ satisfies $\mu  \in (0,{1}/{3T})$.
\end{assumption}

 We have a few comments on Assumption \ref{assump:mu}.
\begin{remark}\label{rmk:correlation1}
First, if the group sizes do not vary much, then the condition $\mu < {1}/{3T}$ holds if $
\nu < 1/C\|x^{\dag}\|_{\ell^0}$. The latter condition with $C\in(2,7)$ is widely used for analyzing lasso
\cite{zhang:2009sharp} and OMP \cite{Tropp:2004,CaiWang:2011}. Hence, the condition in
Assumption \ref{assump:mu} generalizes the classical one. Second,
it allows strong inner-group correlations (i.e., ill-conditioning of
$\Psi_{G_i}$), for which the MC $\nu$ can be very close to one,
and thus it has a built-in mechanism to tackle inner-group correlation. This differs
essentially from existing approaches,
which rely on certain pre-processing techniques \cite{Buhlmann:2013,Witten:2014}.
\end{remark}

\begin{remark}\label{rmk:correlation2}
A similar block MC, defined by $\mu_B = \max_{i\neq j} \|\Psi_{G_i}^{t}\Psi_{G_j}\|/s$,
was used for analyzing group greedy algorithms \cite{Eldar:2010,BenElda:2011} and group
lasso \cite{Bajwa:2015} (without scaling $s$). If every submatrix $\Psi_{G_i}$ is
column orthonormal, i.e., $\Psi_{G_i}^t\Psi_{G_i} = I$, then $\mu_B$ and $\mu$ are identical.
However, to obtain the error estimates in \cite{Eldar:2010,BenElda:2011}, the MC $\nu$
within each group is still needed, which excludes inner-group correlations.
The estimates in \cite{Bajwa:2015} were obtained under the assumption $
\max_{i} \|\Psi_{G_i}^{t}\Psi_{G_i}-I\| \leq 1/{2}$, which again implies that $\Psi_{G_i}$
are well conditioned \cite[Theorem 1]{Bajwa:2015}.
Group restrict eigenvalue conditions \cite{HuangZhang:2010,LouniciPontil:2011} and  group RIP
\cite{ErenVidyasagar:2015} were adopted for analyzing the group lasso. Under these
conditions, strong correlation within groups is also not allowed.
\end{remark}

Now we give a few useful estimates. The proofs can be found in Appendix
\ref{app:G}.
\begin{lemma}\label{lem:est-G}
For any $i,j$, there hold
\begin{equation*}
  \begin{aligned}
   &\|\bar{\Psi}_{G_i}^{-1}\Psi^t_{G_i} y\| \leq \|y\|, \quad \|\Psi_{G_i}\bar{\Psi}_{G_i}^{-1} x_{G_i}\| =\|x_{G_i}\|,\\
   &\|D_{i,j}x_{G_j}\| \left\{\begin{array}{ll} \leq  \mu \|x_{G_j}\|& i \neq j,\\ = \|x_{G_j}\| & i=j .\end{array} \right.
  \end{aligned}
\end{equation*}
\end{lemma}

\begin{lemma}\label{lem:est-D}
For any distinct groups $G_{i_1},\cdots,G_{i_M}$, $1\leq M\leq T$, let
\begin{equation*}
D = \left(\begin{array}{ccc}
D_{i_1,i_1}& \cdots & D_{i_1,i_M} \\
\vdots & \vdots & \vdots \\
D_{i_M,i_1} & \cdots & D_{i_M,i_M}
\end{array}\right)\quad\mbox{ and }\quad x = \left(\begin{array}{c}
x_{G_{i_1}} \\ \vdots \\ x_{G_{i_M}}
\end{array}\right).
\end{equation*}
There holds
$
\|Dx \|_{\ell^\infty(\ell^2)} \in [(1 - (M-1)\mu)\|x\|_{\ell^\infty(\ell^2)},(1 + (M-1)\mu)\|x\|_{\ell^\infty(\ell^2)}].
$
\end{lemma}

Lemma \ref{lem:est-D} directly implies the uniqueness of the oracle solution $x^o$;
see Appendix \ref{app:oracle-unique} for the proof.
\begin{corollary}\label{cor:oracle}
If Assumption \ref{assump:mu} holds, then $x^o$ is unique.
\end{corollary}

\section{Theory of the $\ell^0(\ell^2)$ optimization problem}\label{sec:cwm}
Now we analyze the model \eqref{eqn:groupl0},
e.g., existence of a global minimizer, invariance property, support recovery, and
block coordinatewise minimizers.

\subsection{Existence and property of a global minimizer}
First we show the existence of a global minimizer to problem \eqref{eqn:groupl0}; see
Appendix \ref{app:existence} for the proof.
\begin{theorem}\label{thm:existence}
There exists a global minimizer to problem \eqref{eqn:groupl0}.
\end{theorem}

It can be verified directly that the  $\ell^0(\ell^2)$ penalty is invariant under group full-rank column
transformation, i.e., $\|\bar\Psi_{G_i} x_{G_i}\|_{\ell^0}=\|x_{G_i}\|_{\ell^0}$ for
nonsingular $\bar\Psi_{G_i}$, $i=1,2,\ldots,N$. Thus problem
\eqref{eqn:groupl0} can be equivalently transformed into
\begin{equation}\label{eqn:groupl0-orth}
  \tfrac{1}{2}\|\sum_{i=1}^N\Psi_{G_i}\bar{\Psi}_{G_i}^{-1} \bar{x}_{G_i} -y\|^2 + \lambda \|\bar x\|_{\ell^0(\ell^2)}.
\end{equation}
with $\bar{x}_{G_i} = \bar{\Psi}_{G_i}x_{G_i}$.
This invariance does not hold for other group sparse penalties, e.g., group lasso and group MCP.
Further, the BMC $\mu$ is invariant under the transformation, since $\mathrm{span}
(\{\psi_l: l\in G_i\}) =\mathrm{span}(\{(\Psi_{G_i}\bar \Psi_{G_i}^{-1})_l\})$.

\begin{remark}\label{rmk:correlation3}
Most existing approaches do not distinguish inner- and inter-group columns, and thus
require incoherence between the columns within each group in the theoretical analysis. For strong inner-group
correlation, a clustering step is often employed to decorrelate $\Psi$
\cite{Buhlmann:2013,Witten:2014}. In contrast, our approach has a built-in decorrelation mechanism: it is
independent of the conditioning of the submatrices $\{\Psi_{G_i}\}_{i=1}^N$.
\end{remark}

For a properly chosen $\lambda$,
a global minimizer has nice properties, e.g., exact support recovery for small noise and
oracle property; the proof is given in Appendix \ref{app:oracle}.
\begin{theorem}\label{thm:oracle}
Let Assumption \ref{assump:mu} hold, $x$ be a global minimizer of \eqref{eqn:groupl0} with an
active set $\mathcal{A}$, and $\bar{x}_{G_i}^\dag = \bar{\Psi}_{G_i}x_{G_i}^\dag$.
\begin{itemize}
\item[(i)] Let $\Lambda = |\{i\in \mathcal{A}^\dag: \|\bar{x}^\dag_{G_i}\|< 2\sqrt{2\lambda} + 3\epsilon\}|$. If $\lambda > {\epsilon^2}/{2}$, then
$|\calA\setminus\calA^\dag| + |\calA^\dag\setminus\calA| \leq 2 \Lambda.$
\item[(ii)] If $\eta$ is small, i.e., $\epsilon < \min_{i\in \mathcal{A}^\dag}\{\|\bar{x}^\dag_{G_i}\|\}/5,$
then for any $\lambda \in ({\epsilon^2}/{2}, (\min_{i\in \mathcal{A}^\dag}\{\|\bar{x}^\dag_{G_i}\|\} - 2\epsilon)^2/8)$, the oracle
solution $x^o$ is the only global minimizer to $J_\lambda$.
\end{itemize}
\end{theorem}

\subsection{Necessary optimality condition}\label{ssec:bcwm}
Since problem \eqref{eqn:groupl0} is highly nonconvex, there seems no convenient characterization of a
global minimizer that is amenable with numerical treatment.
Hence, we resort to the concept of a block coordinatewise minimizer (BCWM) with respect
to the group partition $\{G_i\}_{i=1}^N$, which is
minimizing along each group coordinate $x_{G_i}$ \cite{Tseng:2001}. Specifically,
a BCWM $x^*$ to the functional $J_\lambda$ satisfies for $i=1,2,\ldots,N$
\begin{equation*}
x_{G_i}^* \in \arg\min _{x_{G_i}\in\mathbb{R}^{s_i}} J_\lambda(x_{G_1}^*,\cdots,x_{G_{i-1}}^*, x_{G_i},x_{G_{i+1}}^*,\cdots,x_{G_N}^*).
\end{equation*}

We have the following necessary and sufficient condition for a BCWM $x^*$; see
Appendix \ref{app:bcwm} for the proof.
It is also the necessary optimality condition of a global minimizer $x^*$.
\begin{theorem}\label{prop:necopt}
The necessary and sufficient optimality condition for a BCWM $x^*\in \mathbb{R}^p$ of problem \eqref{eqn:groupl0} is given by
\begin{equation}\label{eqn:opt}
  \bar{x}_{G_i}^* \in H_\lambda(\bar x_{G_i}^* + \bar d^*_{G_i}),\quad i=1,\ldots,N,
\end{equation}
where $\bar x_{G_i}^* = \bar\Psi_{G_i} x_{G_i}^*$, and the dual variable $d^*$ is $d^*=\Psi^t(y-\Psi x^*)$ and $\bar d^*_{G_i}=\bar\Psi_{G_i}^{-1} d_{G_i}^*$.
\end{theorem}

\begin{remark}
The optimality system is expressed in terms of the transformed variables $\bar x$ and
$\bar d$ only, instead of the primary variables $x$ and $d$. This has important
consequences for the analysis and algorithm of the $\ell^0(\ell^2)$ model: both should
be carried out in the transformed domain. Clearly, \eqref{eqn:opt} is also the
optimality system of a BCWM $\bar x^*$ for problem \eqref{eqn:groupl0-orth}, concurring
with the invariance property.
\end{remark}
\noindent\textbf{Notation.}
In the discussions below, given a primal variable $x$ and dual variable $d$, we will
use $(\bar{x},\bar{d})$ for the transformed variables, i.e., $\bar{x}_{G_i} = \bar{\Psi}_{G_i}x_{G_i}$
and $\bar{d}_{G_i} = \bar{\Psi}_{G_i}^{-1}d_{G_i}$, $i=1,...,N$.

Using the group hard-thresholding operator $H_\lambda$, we deduce
\begin{equation*}
\begin{array}{l}
  \|\bar x_{G_i}^* + \bar d_{G_i}^*\|<\sqrt{2\lambda} \Rightarrow \bar x_{G_i}^*=0\ \  (\Leftrightarrow x_{G_i}^*=0),\\[1.2ex]
 \|\bar x_{G_i}^*+ \bar d_{G_i}^*\|>\sqrt{2\lambda} \Rightarrow \bar d_{G_i}^*= 0\ \ (\Leftrightarrow d_{G_i}^* = 0).
 \end{array}
\end{equation*}
Combining these two relations gives a simple observation
\begin{equation}\label{eqn:xtd}
  \|\bar x_{G_i}\|\geq \sqrt{2\lambda}\geq \|\bar d_{G_i}\|.
\end{equation}

Next we discuss interesting properties of a BCWM $x^*$. First, it is
always a local minimizer, i.e., $J_\lambda(x^*+h)\geq J_\lambda(x^*)$ for all small $h\in\mathbb{R}^p$;
see Appendix \ref{app:bwcm} for the proof.
\begin{theorem}\label{thm:local}
A BCWM $x^*$  of the functional $J_\lambda$ is a local minimizer. Further, with its active set $\mathcal{A}$,
if $\Psi_{\cup_{i\in\mathcal{A}}G_i}$ has full column rank, then it is a strict local minimizer.
\end{theorem}

Given the active set $\calA$ of a BCWM $x^*$, if $|\mathcal{A}|$ is controlled,
then $\mathcal{A}$ provides information about  $\calA^\dag$; see
Theorem \ref{thm:support} below and Appendix \ref{app:support} for
the proof. In particular, if the noise $\eta$ is small, with a
proper choice of $\lambda$,  then $\mathcal{A}\subseteq \mathcal{A}^\dag$.
\begin{theorem}\label{thm:support}
Let Assumption \ref{assump:mu} hold, and $x^*$ be a BCWM to the model \eqref{eqn:groupl0} with a
support $\calA$ and $|\calA|\leq T$. Then the following statements hold.
\begin{itemize}
\item[(i)] The inclusion $\{i: \|\bar{x}^\dag_{G_i}\| \geq 2\sqrt{2\lambda} + 3\epsilon\}\subseteq \calA$ holds.
\item[(ii)] The inclusion $\calA\subseteq \calA^\dag$ holds if $\epsilon$ is small:
\begin{equation}\label{assump:noise}
\epsilon \leq t\min_{i\in \mathcal{A}^\dag}\{\|\bar{x}^\dag_{G_i}\|\} \mbox{ for some } 0\leq t < \tfrac{1-3\mu T}{2}.
\end{equation}
\item[(iii)] If the set $\{i\in\mathcal{A}^\dag: \|\bar{x}^\dag_{G_i}\| \in [2\sqrt{2\lambda} - 3\epsilon, 2\sqrt{2\lambda} + 3\epsilon]\}$ is empty, then $\mathcal{A}\subseteq\mathcal{A}^\dag$.
\end{itemize}
\end{theorem}

\section{Group Primal-Dual Active Set Algorithm}\label{sec:pdasc}

Now we develop an efficient, accurate and globally convergent group primal dual active set with continuation (GPDASC) algorithm for
 problem \eqref{eqn:groupl0}. It generalizes the algorithm for the $\ell^1$ and $\ell^0$
regularized problems \cite{FanJiaoLu:2014,JiaoJinLu:2014} to the group case.

\subsection{GPDASC algorithm}
The starting point is the necessary
and sufficient optimality condition \eqref{eqn:opt} for a BCWM $x^*$, cf. Theorem \ref{prop:necopt}. The
following two observations from \eqref{eqn:opt} form the basis of the derivation. First, given a BCWM $x^*$ (and
its dual variable $d^*=\Psi^t(y-\Psi x^*)$), one can determine the active set $\calA^*$ by
\begin{equation*}
  \calA^* = \{i: \|\bar x_{G_i}^* + \bar d_{G_i}^*\|>\sqrt{2\lambda}\}
\end{equation*}
and the inactive set $\calI^*$ its complement, provided that the set $\{i: \|\bar x_{G_i}^*+
\bar d_{G_i}^*\|=\sqrt{2\lambda}\}$ is empty. Second, given the active set $\calA^*$, one
can determine uniquely the primal and dual variables $x^*$ and $d^*$ by (with $B=\cup_{i\in\calA^*}G_i$)
\begin{equation*}
   \left\{\begin{aligned} &x^*_{G_i} = 0\;\,\forall i\in \calI^*\quad \mbox{and} \quad\Psi_B^t \Psi_{B} x_{B}^* = \Psi_{B}^t y,\\
       & d^*_{G_j} = 0\;\,\forall j \in\calA^*\quad \mbox{and}\quad  d^*_{G_i} = \Psi_{G_i}^t(y-\Psi x^*)\;\,\forall i\in\calI^*.
   \end{aligned}\right.
\end{equation*}
By iterating these two steps alternatingly, with the current estimates $(x,d)$ and $(\calA,\calI)$
in place of $(x^*,d^*)$ and $(\calA^*,\calI^*)$, we arrive at an algorithm for
problem \eqref{eqn:groupl0}.

The complete procedure is listed in Algorithm \ref{alg:gpdasc}. Here $K_{max}\in\mathbb{N}$
is the maximum number of inner iterations, $\lambda_0$ is the initial guess of $\lambda$.
The choice $\lambda_0=\frac{1}{2}\|y\|^2$ ensures that $x^0=0$ is the only global minimizer, cf. Proposition
\ref{prop:lam0} below, with a dual variable $d^0=\Psi^ty$. The scalar $\rho\in(0,1)$ is the decreasing factor for
$\lambda$, which essentially determines the length of the continuation path.

\begin{algorithm}
   \caption{GPDASC algorithm}\label{alg:gpdasc}
   \begin{algorithmic}[1]
     \STATE Input: $\Psi\in \mathbb{R}^{n\times p}$, $\{G_i\}_{i=1}^N$, $K_{max}$, $\lambda_0=\frac{1}{2}\|y\|^2$, and $\rho\in (0,1)$.
     \STATE Compute $\bar{\Psi}_{G_i} = (\Psi_{G_i}^t\Psi_{G_i})^{1/2}$.
     \STATE Set $x(\lambda_0)=0$, $d(\lambda_0) = \Psi^t y$, $\mathcal{A}(\lambda_0) = \emptyset$.
     \FOR {$s=1,2,...$}
     \STATE Set $\lambda_s = \rho\lambda_{s-1}$, $x^0 = x(\lambda_{s-1})$, $d^0=d(\lambda_{s-1})$, $\mathcal{A}_{-1}=\mathcal{A}(\lambda_{s-1})$.
      \FOR {$k=0,1,\ldots,K_{max}$} %
     \STATE Let $\bar{x}^k_{G_i} = \bar{\Psi}_{G_i}x^k_{G_i}$ and $\bar{d}^k_{G_i} = \bar{\Psi}_{G_i}^{-1}d^k_{G_i}$, and define $${\calA}_k = \{i: \|\bar{x}^k_{G_i} + \bar{d}^k_{G_i}\|> \sqrt{2\lambda_s}\}.$$
     \STATE Check the stopping criterion $\mathcal{A}_k = \mathcal{A}_{k-1}$.
     \STATE Update the primal variable $x^{k+1}$ by
     \begin{equation*}
        x^{k+1} = \mathop\textrm{argmin}\limits_{\textrm{supp}(x) \subseteq \cup_{i\in {\calA}_k} G_{i}} \|\Psi x - y\|.
     \end{equation*}
     \STATE Update the dual variable by $d^{k+1} = \Psi^t(y -\Psi x^{k+1})$.
     \ENDFOR
     \STATE Set the output by $x(\lambda_s)$, $d(\lambda_s)$ and $\mathcal{A}(\lambda_s)$.
     \STATE Check the stopping criterion
      \begin{equation}\label{eqn:discprin}
      \|\Psi x(\lambda_s) - y\| \leq \epsilon.
      \end{equation}
     \ENDFOR
   \end{algorithmic}
\end{algorithm}

The algorithm consists of two loops: an inner loop of solving problem \eqref{eqn:groupl0}
with a fixed $\lambda$ using a GPDAS algorithm (lines 6--10),
and an outer loop of continuation along
the parameter $\lambda$ by gradually decreasing its value.

In the inner loop, it involves a least-squares problem:
\begin{equation*}
   x^{k+1} = \mathop\textrm{argmin}\limits_{\textrm{supp}(x) \subseteq \cup_{i\in {\calA}_k} G_{i}} \|\Psi x - y\|,
\end{equation*}
which amounts to solving a (normal) linear system of size $|\cup_{i\in\calA_k}G_i|\leq |\calA_k|s$.
Hence, this step is very efficient, if the active set $\calA_k$ is small, which is the case for
group sparse signals. Further, since the inner iterates are of Newton type \cite{FanJiaoLu:2014},
the local convergence should be  fast. However, in order to fully exploit this nice feature, a good initial guess
of the primal and dual variables $(x,d)$ is required. To this end, we apply a continuation strategy along
$\lambda$. Specifically, given a large $\lambda_0$, we gradually decrease its
value by $\lambda_s=\rho\lambda_{s-1}$, for some decreasing factor $\rho\in(0,1)$, and take the solution
$(x(\lambda_{s-1}),d(\lambda_{s-1}))$ to the $\lambda_{s-1}$-problem $J_{\lambda_{s-1}}$ to warm start the
$\lambda_s$-problem $J_{\lambda_s}$.

There are two stopping criteria in the algorithm, at steps 8 and 13, respectively. In the inner loop,
one may terminate the iteration if the active set $\calA_k$ does not change or a maximum number $K_{max}$ of
inner iterations is reached. Since the stopping criterion $\calA_{k}=\calA_{k-1}$ for convex
optimization may never be reached in the nonconvex context \cite{JiaoJinLu:2014}, it has to be terminated after a maximum
number $K_{max}$ of iterations. Our convergence analysis holds for any $K_{max}\in\mathbb{N}$,
including $K_{max} =1$, and we recommend  $K_{max}\leq 5$ in practice. The stopping criterion at
step 13 is essentially concerned with the proper choice of $\lambda$. The choice of
$\lambda$ stays at the very heart of the model \eqref{eqn:groupl0}.
Many rules, e.g., discrepancy principle, balancing principle and information
criterion, have been developed for variational regularization \cite{ItoJin:2014}. In Algorithm \ref{alg:gpdasc}, we give only
the discrepancy principle \eqref{eqn:discprin}, assuming that a reliable estimate on the noise level $\epsilon$ is
available. The rationale behind the principle is that the reconstruction accuracy should
be comparable with the data accuracy. Note that the use of \eqref{eqn:discprin}
(and other rules) does not incurred any extra computational overheads,
since the sequence of solutions $\{x(\lambda_s)\}$ is already generated along the continuation path.

Now we justify the choice of $\lambda_0$: for
large $\lambda$, $0$ is the only global minimizer to $J_\lambda$.
The proof is given in Appendix \ref{app:lam0}.
\begin{proposition}\label{prop:lam0}
The following statements hold.
\begin{itemize}
  \item[(i)] For any $\lambda>0$, $x^*=0$ is a strict local minimizer to $J_\lambda$;
  \item[(ii)] For any $\lambda>\lambda_0:=\tfrac{1}{2}\|y\|^2$, $x^*=0$ is the only global minimizer of
problem \eqref{eqn:groupl0}.
\end{itemize}
\end{proposition}

\subsection{Convergence analysis}\label{ssec:conv}

Now we state the global convergence of Algorithm \ref{alg:gpdasc}.
\begin{theorem}\label{thm:main}
Let Assumption \ref{assump:mu} and \eqref{assump:noise} hold. Then for a proper choice of $\rho\in (0,1)$,
and for any $K_{max}\geq 1$, Algorithm \ref{alg:gpdasc} converges to $x^o$ in a finite number of iterations.
\end{theorem}

We only sketch the main ideas, and defer the lengthy proof to Appendix \ref{app:07}.
The most crucial ingredient of the proof is to characterize
a  monotone decreasing property of the ``energy'' during the iteration by some auxiliary
set $\Gamma_s$ defined by
\begin{equation}\label{equ:indexset}
\Gamma_{s} = \left \{i: \|\bar{x}^\dag_{G_i}\| \geq \sqrt{2s}\right\}.
\end{equation}
The inclusion $\Gamma_{s_1}\subseteq\Gamma_{s_2}$ holds trivially for $s_1>s_2$.
If $\calA_k$ is the active set at the $k^{\rm th}$ iteration,
the corresponding energy $E_k$ is defined by
$E_k = E(\mathcal{A}_k) = \max_{i\in\mathcal{I}_k} \|\bar{x}^\dag_{G_i}\|.$
Then with properly chosen $s_1 > s_2$, there holds
$
\Gamma_{s_1^2 \lambda} \subseteq \calA_k\subseteq \calA^\dag \Rightarrow \Gamma_{s_2^2
\lambda} \subseteq \calA_{k+1}\subseteq \calA^\dag.
$
This relation characterizes the evolution of the active set $\calA_k$,
 and provides a crucial strict monotonicity of
the energy $E_k$. This observation is sufficient to show the convergence of the algorithm to the oracle
solution $x^o$ in a finite number of steps; see Appendix \ref{app:07} for details.

\begin{remark}
The convergence in Theorem \ref{thm:main} holds for any $K_{\max}\in\mathbb{N}$, including
$K_{\max}=1$. According to the proof in Appendix \ref{app:07}, the smaller are the factor
$\mu T$ and the noise level $\epsilon$, the smaller is the decreasing factor $\rho$ that
one can choose and thus Algorithm \ref{alg:gpdasc} takes fewer outer iterations to reach
convergence on the continuation path. We often taken $\rho = 0.7$.
\end{remark}

\subsection{Computational complexity}\label{ssec:complexity}
Now we comment on the computational complexity of Algorithm \ref{alg:gpdasc}. First, we
consider one inner iteration. Steps 7-8 take $O(p)$ flops. At Step 9, explicitly forming
the matrix $\Psi_{B_k}^t\Psi_{B_k}$, $B_k=\cup_{i\in \calA_k}G_i$, takes $O(n|B_k|^2)$
flops, and the cost of forming $\Psi^t y$ is negligible since it is often precomputed.
The Cholesky factorization costs $ O(|B_k|^3)$ flops and the back-substitution needs
$O(|B_k|^2)$ flops. Hence step 9 takes $O(\max(|B_k|^3,n|B_k|^2))$ flops. At step 10,
the matrix-vector product takes $O(np)$ flops. Hence, the the overall cost of one
inner iteration is  $O(\max(|B_k|^3,pn,n|B_k|^2))$. Since the GPDAS is of Newton type, a
few iterations suffice convergence, which is numerically confirmed in Section
\ref{sec:numer}. So with a good initial guess, for each fixed $\lambda$, the overall cost
is $O(\max(|B_k|^3,pn,|B_k|^2n))$. In particular, if the true solution $x^\dag$ is
sufficiently sparse, i.e., $|B_{k}|\ll \min(n,\sqrt{p})$, the cost of per inner iteration is $O(np)$.

Generally, one can apply the well-know low-rank Cholesky up/down-date formulas \cite{GillGolub:1974}
to further reduce the cost. Specifically, with $B_{k}=\cup_{i\in \calA_{k}}G_i$,
we down-date by removing the columns in $B_{k-1}$ but not in
$B_k$ at the cost of $ O(|B_{k-1}\setminus B_{k}||B_{k-1}|^2)$ flops, and
update by appending the columns in  $B_k$ but not in $B_{k-1}$
in $ O(|B_{k}\setminus B_{k-1}|(|B_{k-1}|^2+n|B_{k-1}|))$ flops.
Then the Cholesky factor of $\Psi_{B_{k}}^t
\Psi_{B_k}$ is $O((|B_{k-1}\cup B_{k}|-
|B_{k-1}\cap B_{k}|)|B_{k-1}|(n+|B_{k-1}|))$. Along
the continuation path, $(|B_{k-1}\cup B_{k}|-|B_{k-1}\cap B_{k}|)$
is small, as confirmed in Fig. \ref{fig:symdiffas} below, and thus the overall cost is often of $O(np)$.

\section{Numerical results and discussions}\label{sec:numer}
Now we present numerical results to illustrate distinct features of the proposed $\ell^0(\ell^2)$
model and the efficiency and accuracy of Algorithm \ref{alg:gpdasc}. All the numerical experiments were performed on
a four-core desktop computer with 3.16 GHz and 8 GB RAM. The \texttt{MATLAB} code (GPDASC) is available at
\url{ http://www0.cs.ucl.ac.uk/staff/b.jin/software/gpdasc.zip}.

\subsection{Experimental setup}\label{ssec:setup}
First we describe the problem setup of the numerical experiments. In all the numerical examples, the group sparse structure of the
true signal $x^{\dag}$ is encoded in the partition $\{G_i\}_{i=1}^N$, which is of equal
group size $s$, with $p=Ns$, and $x^\dag$ has $T=|\mathcal{A}^\dag|$ nonzero groups.
The dynamic range (DR) of the signal $x^\dag$ is defined by
\begin{equation*}
  \mathrm{DR} = \frac{\max \{|x^{\dag}_{i}|: x^{\dag}_{i} \neq 0\}}{\min \{|x^{\dag}_{i}|: x^{\dag}_{i} \neq 0\}}.
\end{equation*}
We fix the minimum nonzero entry at $\min \{|x^{\dag}_{i}|: x^{\dag}_{i} \neq 0\}=1.$ The sensing matrix
$\Psi$ is constructed as follows. First we generate a random Gaussian matrix $\widetilde{\Psi}\in
\mathbb{R}^{n\times p}$, $n\ll p$, with its entries following an independent identically distributed
(i.i.d.) standard Gaussian distribution with a zero mean and unit variance.
Then for any $i \in\{ 1,2...,N\}$, we introduce correlation within the $i$th group ${G}_i$ by:
given $\overline{\Psi}_{G_{i}}\in\mathbb{R}^{n\times |G_i|}$ by setting $\overline{\psi}_1 =
\widetilde{\psi}_1$, $\overline{\psi}_{|G_i|}= \widetilde{\psi}_{|G_i|}$  and
\begin{equation*}
   \overline{\psi}_j = \widetilde{\psi}_j + \theta(\widetilde{\psi}_{j-1} + \widetilde{\psi}_{j+1}), \ \ j=2,...,|G_i|-1,
\end{equation*}
where the parameter $\theta\geq 0$ controls the degree of inner-group correlation: The larger
is $\theta$, the stronger is the correlation. Finally, we normalize the matrix $\overline{\Psi}$
to obtain $\Psi$ such that all columns are of unit length. The data $y$ is formed by adding
noise $\eta$ to the exact data $y^\dag = \Psi x^{\dag}$ componentwise, where the entries $\eta_i$
follow an i.i.d. Gaussian distribution $N(0,\sigma^2)$. Below we shall denote by the tuple
$(n,p, N,T,s,{\rm DR},\theta,\sigma )$ the data generation parameters, and the notation $N_1:d:N_2$
denotes the sequence of numbers starting with $N_1$ and less than $N_2$ with a spacing $d$.

\subsection{Comparison with existing group sparse models}

\begin{figure}
\centering
   \begin{tabular}{c}
   \includegraphics[trim = 0.5cm 0cm 1cm 0cm, clip=true,width=7cm]{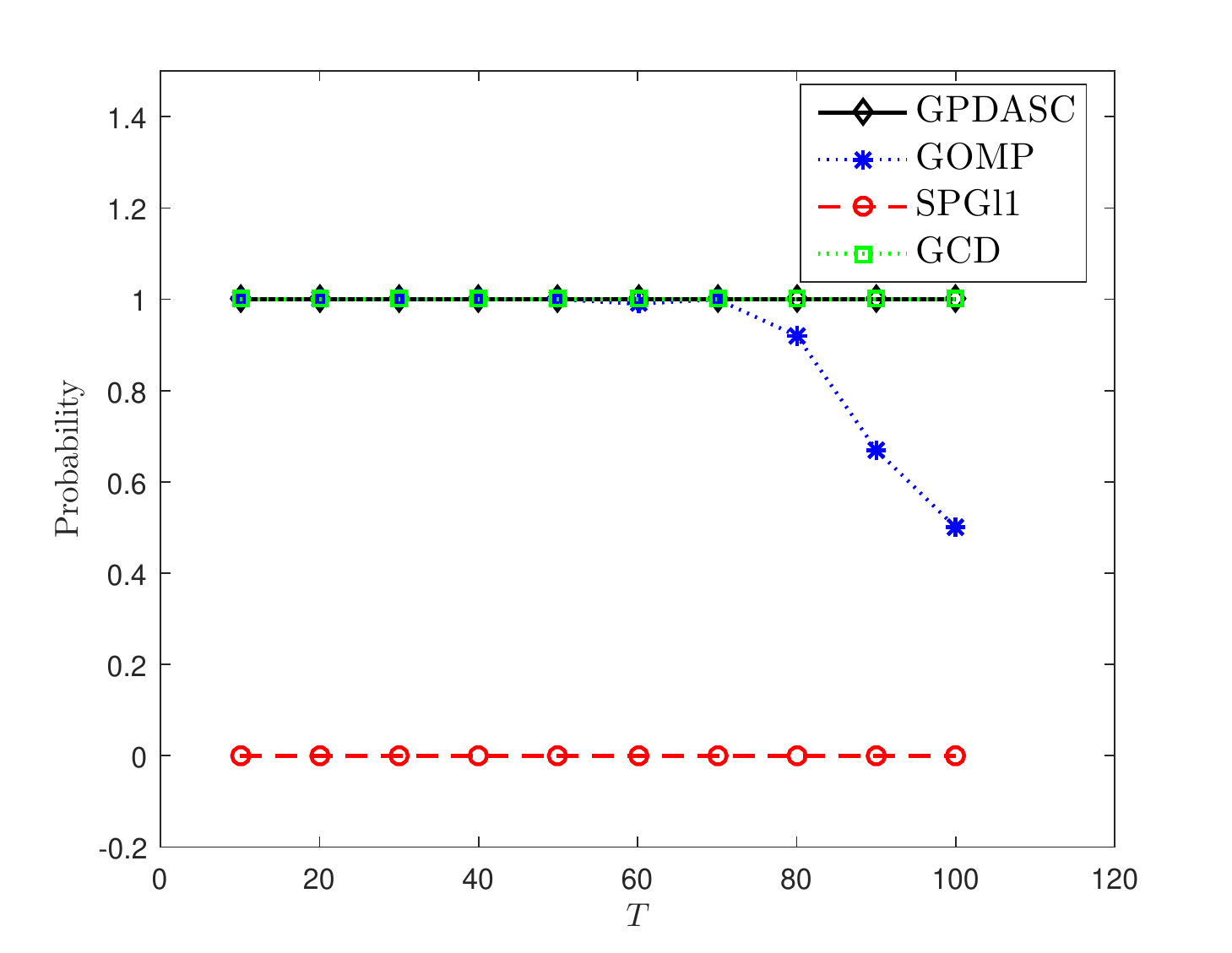} \\
   (a) $(800,2\times10^3,500,10:10:100,4,10,0,10^{-3})$\\
   \includegraphics[trim = 0.5cm 0cm 1cm 0cm, clip=true,width=7cm]{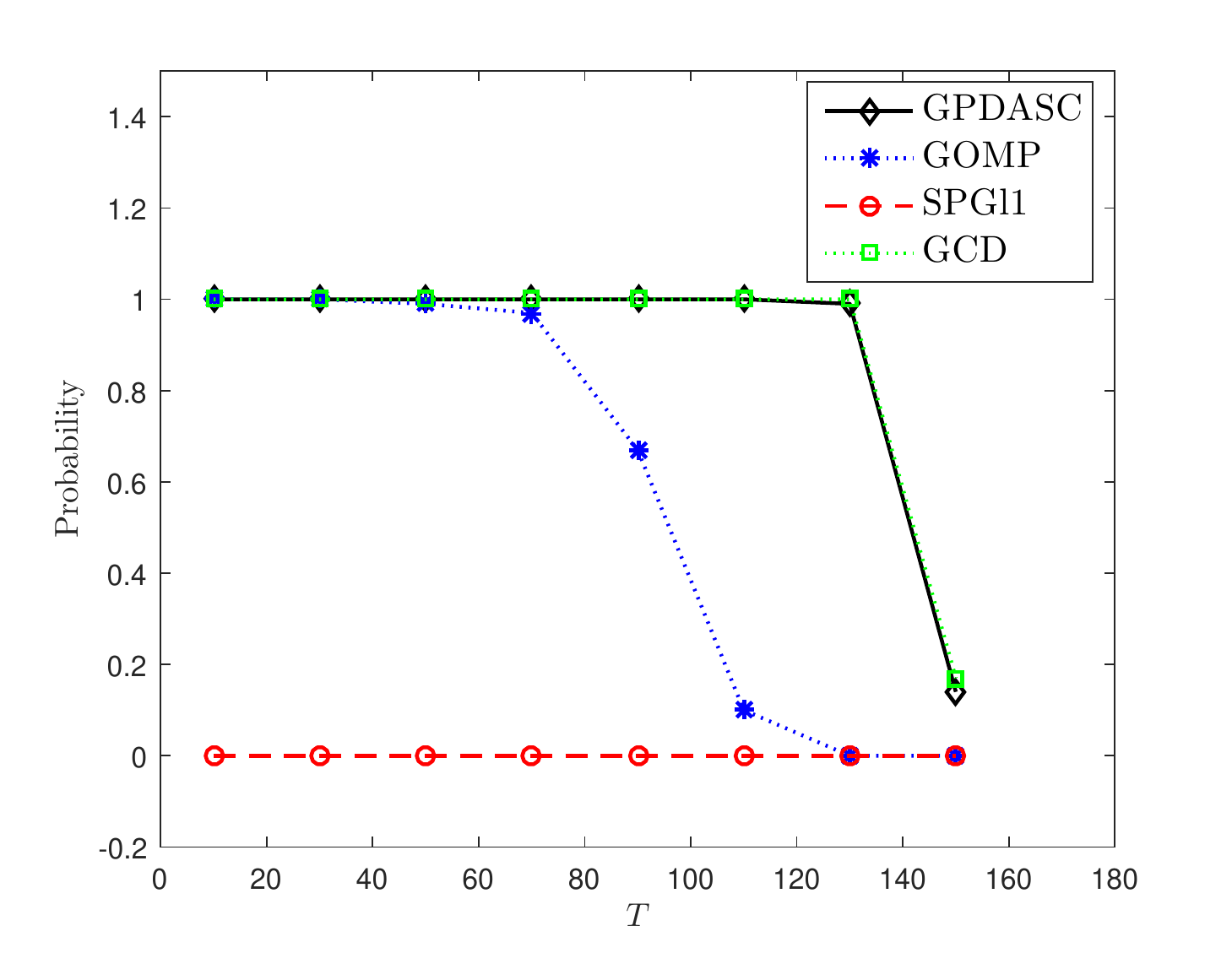}\\
   (b) $(800,2\times10^3,500,10:10:100,4,10,3,10^{-3})$
  \end{tabular}
  \caption{The probability of exact support recovery.}\label{fig:exrgm}
\end{figure}

First we compare our $\ell^0(\ell^2)$ model \eqref{eqn:groupl0} (and Algorithm \ref{alg:gpdasc}) with
three state-of-the-art group sparse recovery models and algorithms, i.e., group lasso model
\begin{equation*}
   \min_{x \in \mathbb{R}^{p}}   \|x\|_{\ell^1(\ell^2)}\quad \textrm{ subject to} \quad \|\Psi x-y\|\leq \epsilon
\end{equation*}
(solved by the group SPGl1 method \cite{VandDenBerg:2008}, available at \url{ http://www.cs.ubc.ca/~mpf/spgl1/},
last accessed on December 23, 2015), group MCP (GMCP) model \cite{WangLiHuang:2008,HuangMaZhang:2009,
HuangBrehenyMa:2012} (solved by a group coordinate descent (GCD) method \cite{BrehenyHuang:2015}),
and group OMP (GOMP) \cite{Eldar:2010,BenElda:2011}. We refer to these references for their
implementation details. Since the algorithm essentially determines the performance of each method,
we shall indicate the methods by the respective algorithms, i.e., SPGl1, GCD, GOMP and GPDASC.
In the comparison, we examine separately support recovery, and computing time and reconstruction error.
All the reported results are the average of 100 independent simulations of the experimental setting.

First, to show exact support recovery, we consider the following two problem settings: $
(800,2\times10^3,500, 10:10:100,4,10,0,10^{-3})$ and $(800,2\times10^3,500,10:10:100,4,10,3,10^{-3})$, for which
the condition numbers of the submatrices $\Psi_{G_i}$ are $O(1)$ and $O(10^2)$, respectively, for the
case $\theta=0$ and $\theta=3$, respectively. Given the group size $s=4$, the condition number $O(10^2)$
is fairly large, and thus the latter is numerically far more challenging than the former. The
numerical results are presented in Fig. \ref{fig:exrgm}, where
the exact recovery is measured by $\mathcal{A}^* = \mathcal{A}^\dag$, with
$\mathcal{A}^\dag$ and $\mathcal{A}^*$ being the true and recovered active sets, respectively.

\begin{figure}[ht!]
\centering
   \begin{tabular}{cc}
   \includegraphics[trim = 0.cm 0cm 0cm 0cm, clip=true,width=7.5cm]{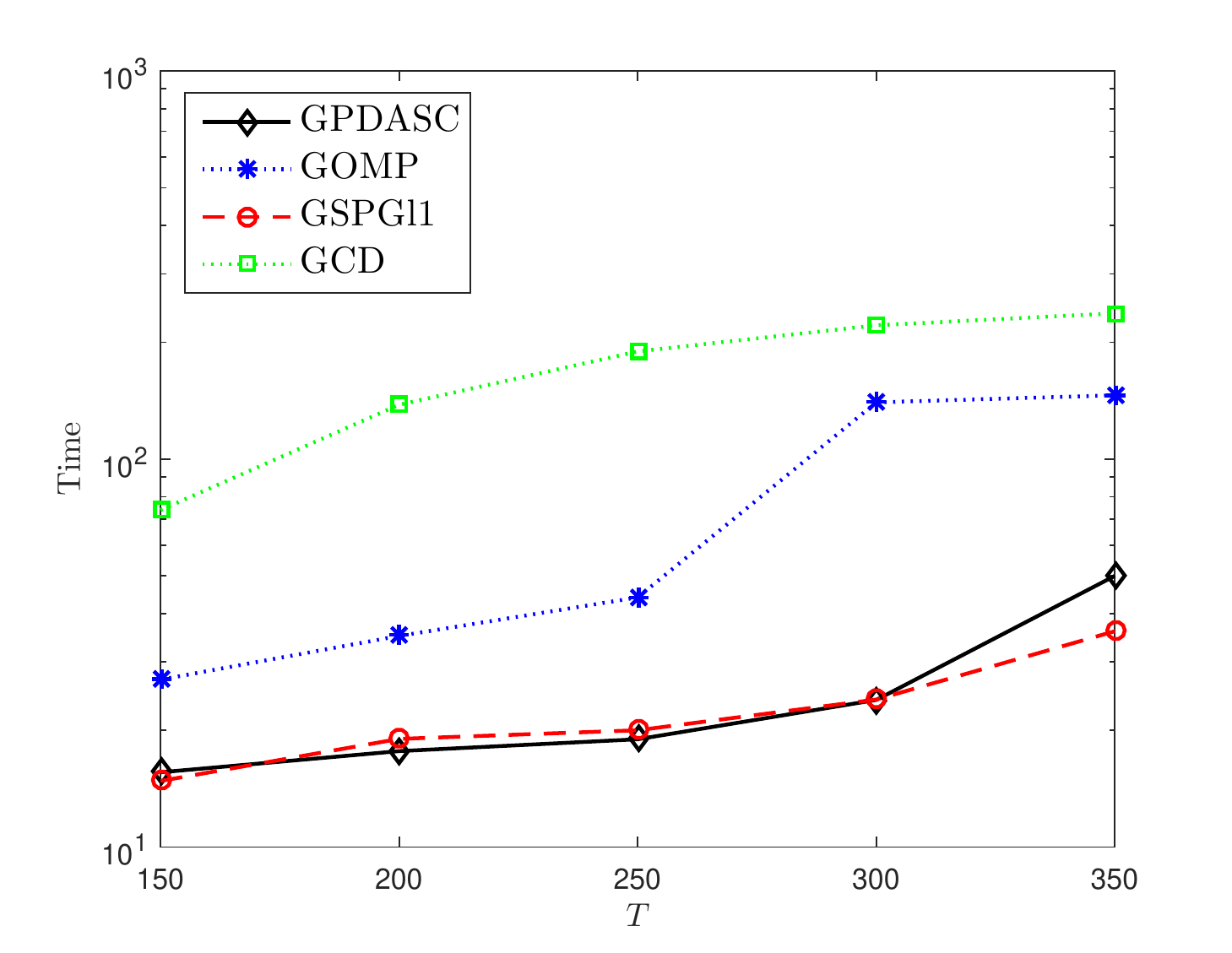} \\
   (a) computing time (in seconds)\\
   \includegraphics[trim = 0cm 0cm 0cm 0cm, clip=true,width=7.5cm]{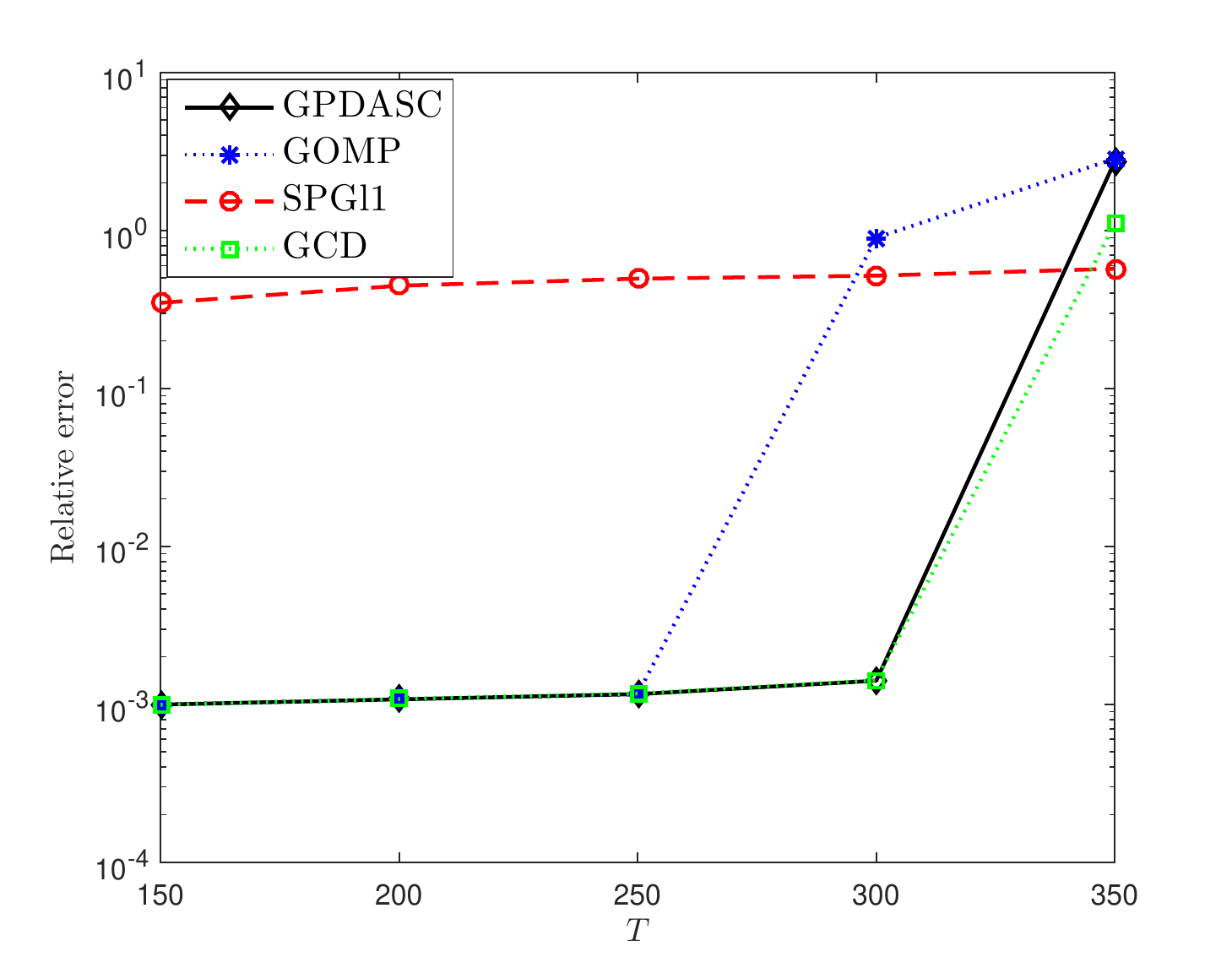}\\
    (b) relative error
  \end{tabular}
  \caption{Computing time and relative error for GPDASC, GOMP,  SPGl1, and GCD for the problem setting $
  (2\times10^3,1\times10^4,2.5\times10^3,150:50:350,4,100,1,10^{-2})$. All computations were performed with the same continuation path.  }
  \label{fig:timeerror1}
\end{figure}

Numerically, it is observed that as the (group) sparsity level $T$ and correlation parameter $\theta$
increase, the $\ell^0(\ell^2)$ model and GMCP are the best performers in the test. Theoretically,
this is not surprising: the $\ell^0(\ell^2)$ model represents the golden-standard for group sparse
recovery, like the $\ell^0$ model for the usual sparsity, and GMCP is a close nonconvex proxy to the
$\ell^0(\ell^2)$ model. Note that GMCP as implemented in \cite{HuangBrehenyMa:2012} is robust with
respect to the inner-group correlation, since it performs a preprocessing step to decorrelate
$\Psi$ by reorthonormalizing the columns within each group. However, unlike the $\ell^0(\ell^2)$
penalty, this step generally changes the GMCP objective function, due to a lack of transform invariance,
and thus may complicate the theoretical analysis of the resulting recovery method.
Meanwhile, as a greedy approximation, GOMP does a fairly good job overall: for small $\theta$,
it can almost perform as well as the $\ell^0(\ell^2)$ model, but deteriorates greatly for large $\theta$.
By its very construction, GOMP from \cite{Eldar:2010} does not take care of the inner-group
correlation directly. Surprisingly, group lasso fails most of the time. A closer look at the recovered signals shows
that it tends to choose a slightly larger active set than $\mathcal{A}^\dag$ in the
noisy case, and this explains its relatively poor performance in terms of
the exact recovery probability, although the relative error is not too large. Intuitively,
this concurs with the fact that the convex relaxation often trades the computational efficiency by
compromising the reconstruction accuracy.

Next we compare their computing time and reconstruction error on the following two problem
settings: $(2\times10^3,1\times10^4,2.5\times10^3,200:25:400,4,100,1,10^{-2})$
and $(5\times10^3,2\times10^4,5\times10^{3},500:50:800,
4, 100,10,10^{-3})$, for which the condition number of the submatrices $\Psi_{G_i}$
is of $O(10)$ and $O(10^3)$, respectively. The case $\theta=10$ involves very strong inner-group
correlation, and it is very challenging. The numerical
results are presented in Figs. \ref{fig:timeerror1} and \ref{fig:timeerror2}.

\begin{figure}[ht!]
\centering
   \begin{tabular}{cc}
   \includegraphics[trim = 0.5cm 0cm 1cm 0cm, clip=true,width=6.85cm]{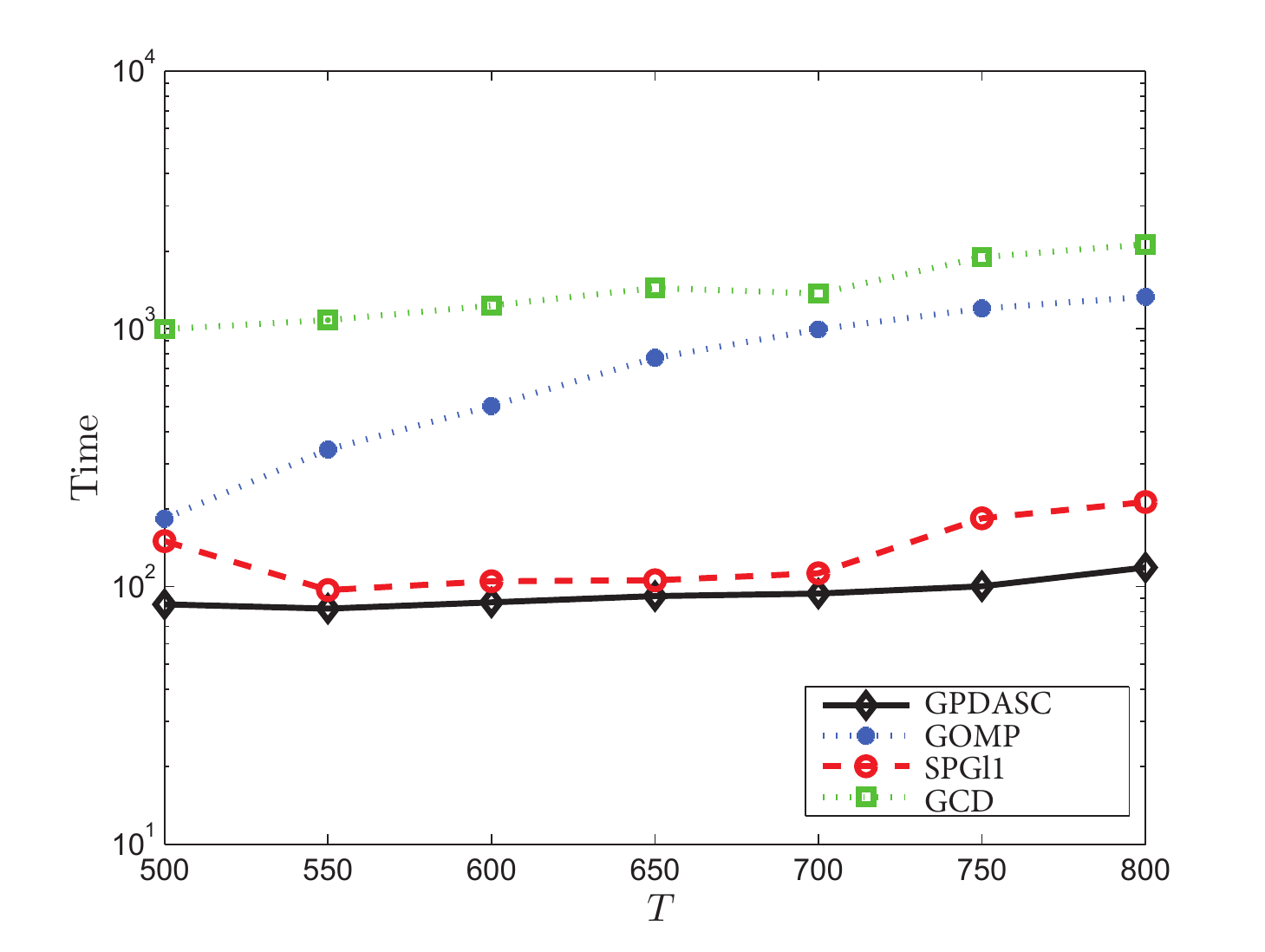} \\
   (a) computing time (in second)\\
   \includegraphics[trim = 0.7cm 0cm 1cm 0cm, clip=true,width=7cm]{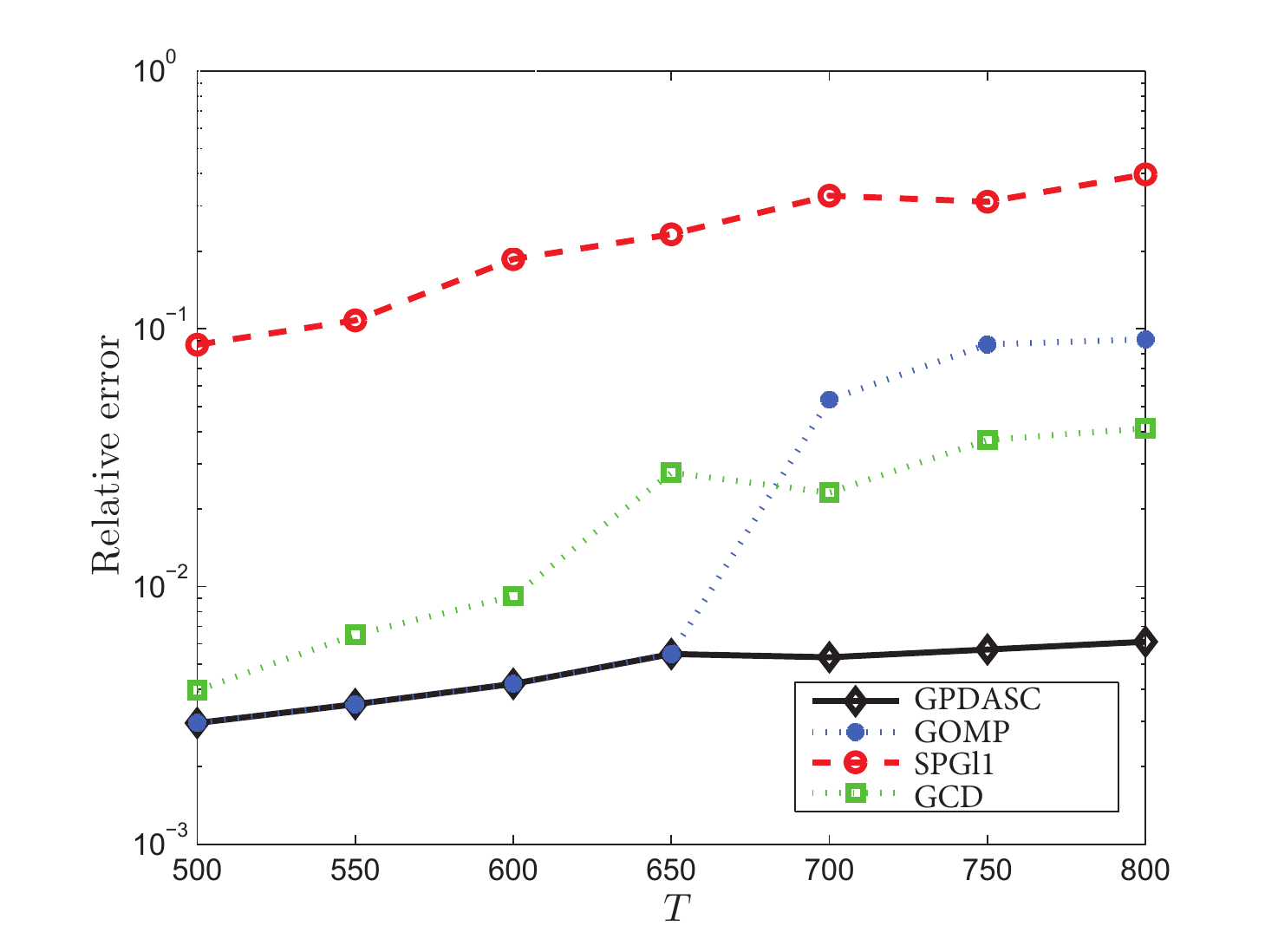}\\
   (b) relative error
  \end{tabular}
  \caption{Computing time and relative error for GPDASC, GOMP,  SPGl1,
   and GCD for the problem setting $(5\times10^3,2\times10^4,5\times10^{-3},500:50:800,4,100,10,10^{-3})$. All computations were performed with the same continuation path.}
  \label{fig:timeerror2}
\end{figure}

For $\theta=1$, the proposed GPDASC for the $\ell^0(\ell^2)$ model is at least three to four times
faster than GCD and GOMP,  cf. Fig. \ref{fig:timeerror1}. The efficiency of GPDASC stems from its
Newton nature and the continuation strategy, apart from solving least-squares problems only on the active
set. We shall examine its convergence more closely below. Group lasso is also computationally
attractive, since due to its convexity, it admits an efficient solver SPGl1. The coupling with a
continuation strategy is beneficial to the efficiency of SPGl1 \cite{FanJiaoLu:2014}. Meanwhile, the
reconstruction errors of the $\ell^0(\ell^2)$ and GMCP are comparable, which is slightly better than GOMP, and
they are much accurate than that of group lasso, as observed earlier. In the case of strong inner-group
correlation (i.e., $\theta=10$), the computing time of GPDASC does not change much, but that of
other algorithms has doubled. Further, the relative error by the $\ell^0(\ell^2)$
model does not deteriorate with the increase of the correlation parameter $\theta$, due to its inherent built-in
decorrelation mechanism, cf. Section \ref{sec:cwm}, and thus it is far smaller than
that by other methods, especially when the group sparsity level $T$ is large. In summary, these experiments show
clearly that the proposed $\ell^0(\ell^2)$ model is very competitive
in terms of computing time, reconstruction error and exact support recovery.

\subsection{Superlinear local convergence of Algorithm \ref{alg:gpdasc}}\label{ssec:iter}

\begin{figure}
\centering
   \begin{tabular}{cc}
   \includegraphics[trim = 0.5cm 0cm 1cm 0cm, clip=true,width=7cm]{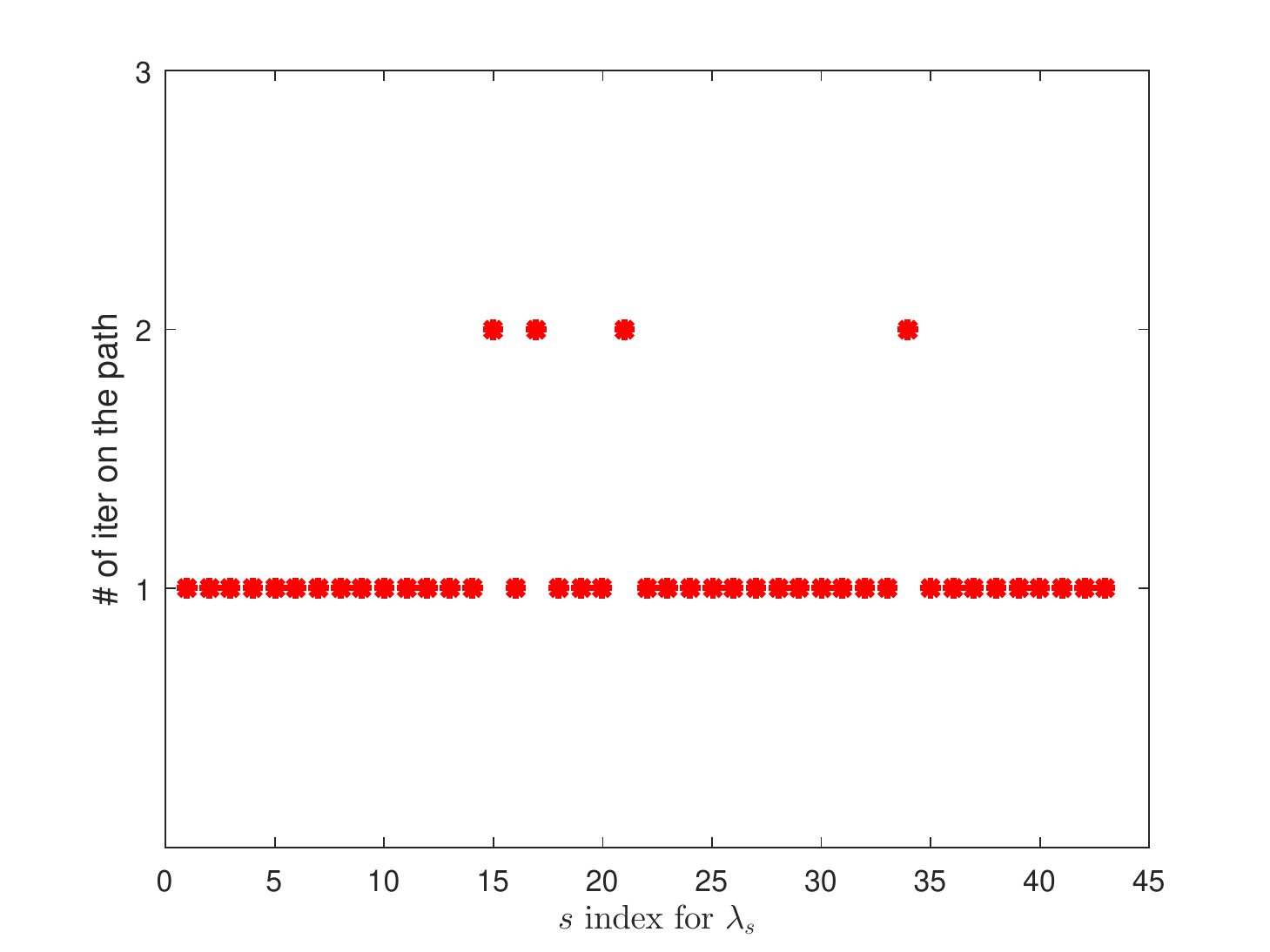} \\
   (a) $(500,10^3,250,50,4,100,0,10^{-3})$ \\
   \includegraphics[trim = 0.5cm 0cm 1cm 0cm, clip=true,width=7cm]{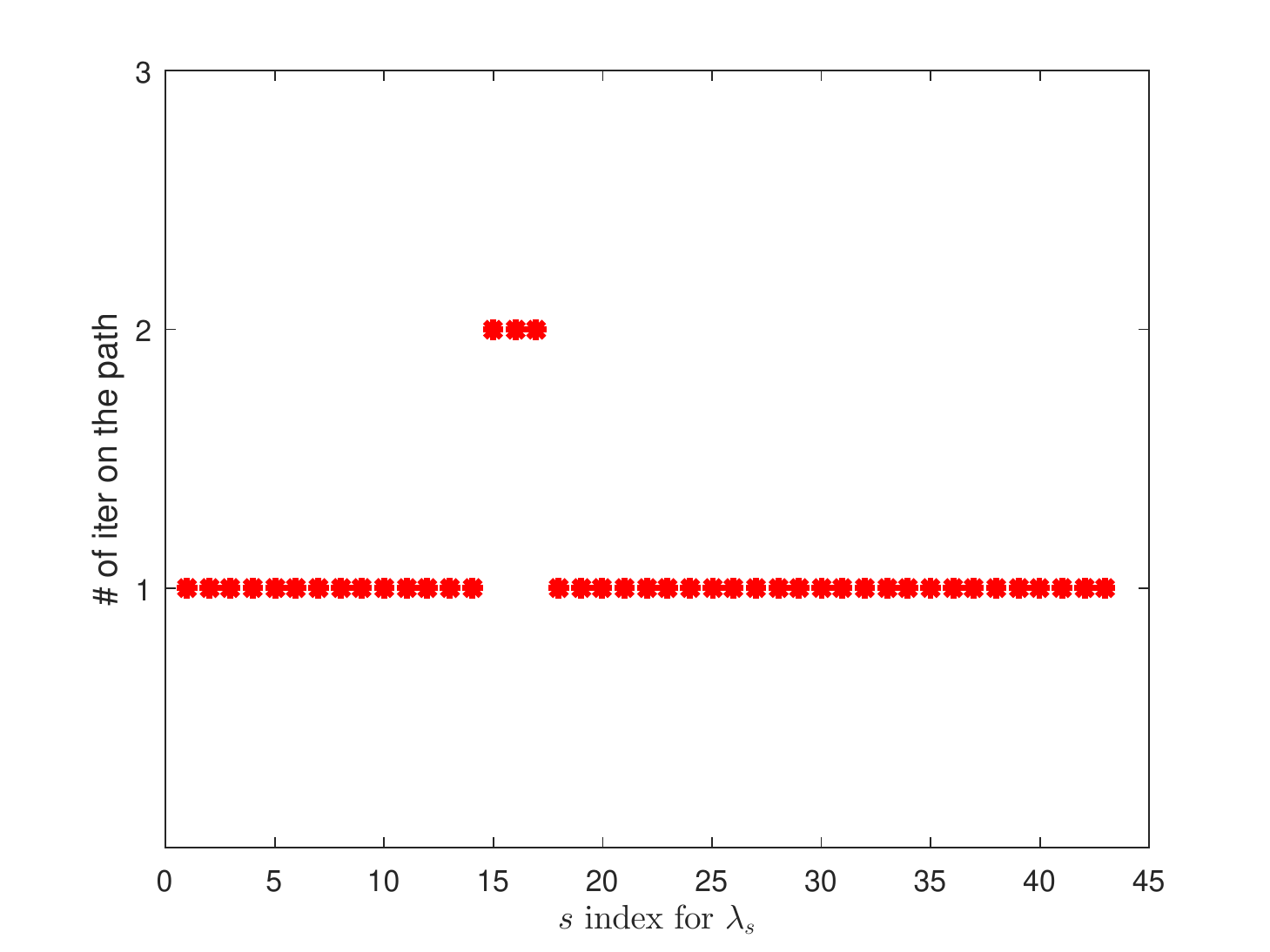}\\
    (b) $(500,10^3,250,50,4,100,3,10^{-3})$
  \end{tabular}
  \caption{The number of iterations along the continuation path, for each fixed regularization parameter
   $\lambda_s$.}\label{fig:locsup}
\end{figure}

We illustrate the convergence behavior of Algorithm \ref{alg:gpdasc} with two problem
settings: $(500,10^3,250,50,4,100,0,10^{-3})$ and $(500,10^3,250,50,4,100,3,10^{-3})$. To
examine the local convergence, we show the number of iterations for each fixed $\lambda_s$
along the continuation path in Fig. \ref{fig:locsup}. It is observed that the stopping
criterion at the inner iteration, i.e., Step 8 of Algorithm \ref{alg:gpdasc}, is usually
reached with one or two iterations, irrespective of the inner-group correlation strength
or the regularization parameter $\lambda_s$. Hence, Algorithm \ref{alg:gpdasc} converges
locally supperlinearly, like that for the convex $\ell^1$ penalty \cite{FanJiaoLu:2014},
and the continuation strategy can provide a good initial guess for each inner iteration
such that the fast local convergence of the GPDAS is fully exploited. This confirms the
complexity analysis in Section \ref{ssec:complexity}. The highly desirable
$\theta$-independence convergence is attributed to the built-in de-correlation effect of
the $\ell^0(\ell^2)$ model.

To gain further insights, we present in Fig. \ref{fig:symdiffas} the variation of the active
set along the continuation path using the setting as that of Fig. \ref{fig:locsup}. It is
observed that the interesting monotonicity relation  $\calA_{s}\subset \calA_{s+1}$ holds along the
continuation path. The difference of active sets between two neighboring regularization parameters $\lambda_s$
is generally small (less than five, and mostly one or two), and thus each GPDAS update is
efficient, with a cost comparable with that of one step gradient descent, if using the
low-rank Cholesky up/down-date \cite{GillGolub:1974}, cf. Section \ref{ssec:complexity}.
Further, the empirical observation that each inner iteration often takes only one
iteration corroborates the convergence theory in Theorem \ref{thm:main}, i.e.,
the algorithm converges globally even if each inner loop takes one iteration.
\begin{figure}
\centering
   \begin{tabular}{cc}
   \includegraphics[trim = 0.5cm 0cm 1cm 0cm, clip=true,width=7cm]{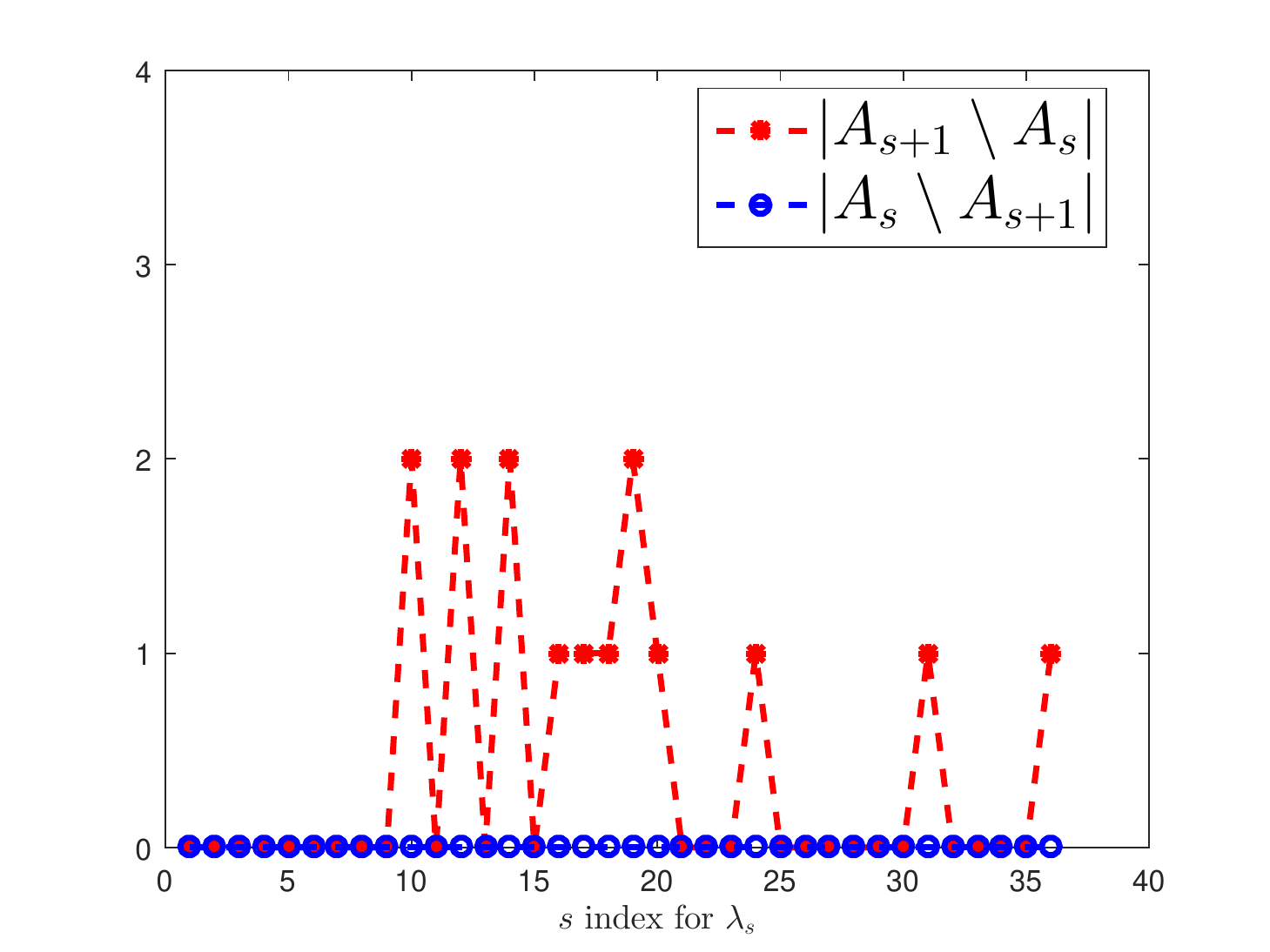} \\
   (a) $(500,10^3,250,15,4,100,0,10^{-3})$ \\
   \includegraphics[trim = 0.5cm 0cm 1cm 0cm, clip=true,width=7cm]{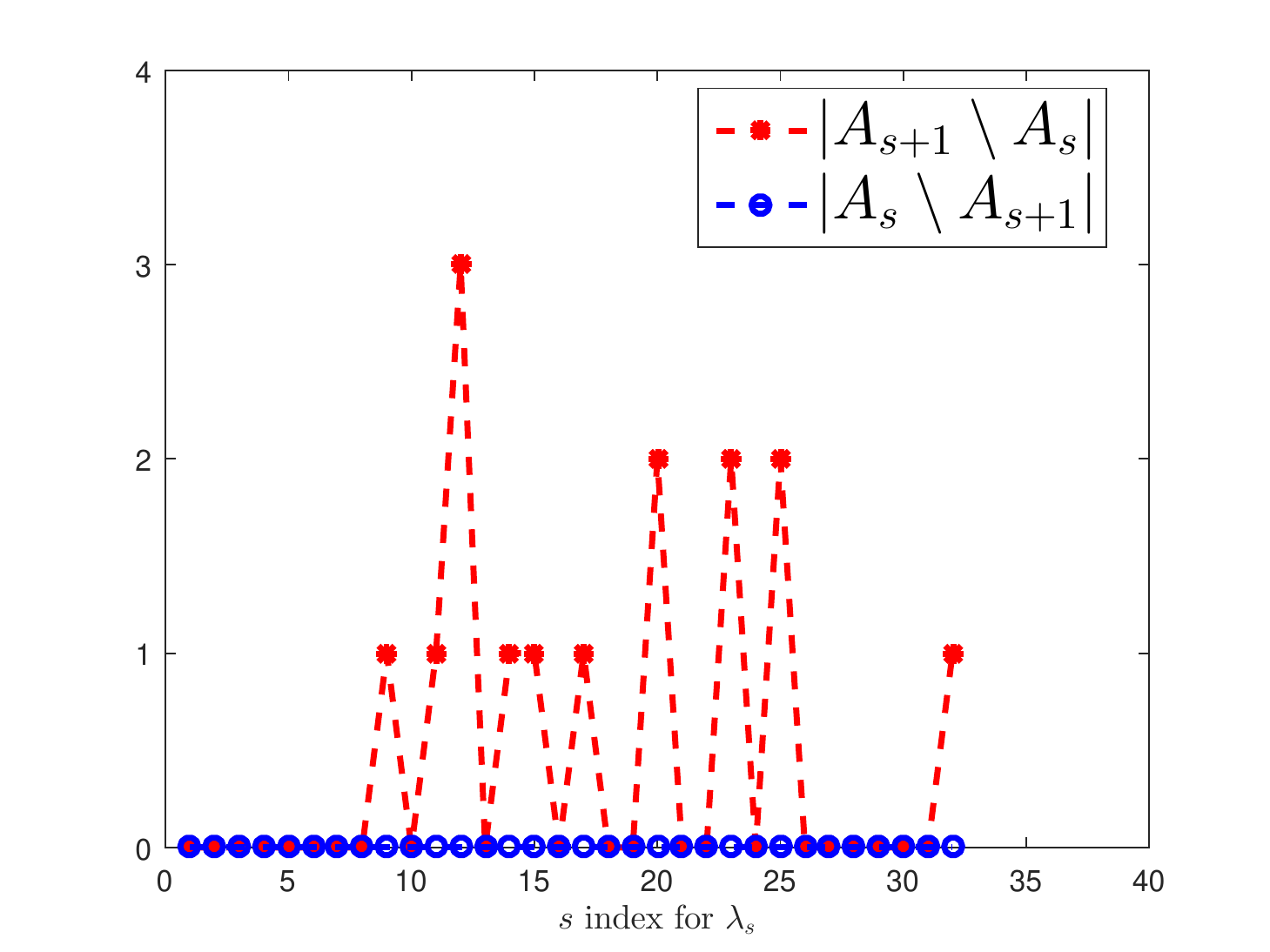}\\
    (b) $(500,10^3,250,15,4,100,1,10^{-3})$
  \end{tabular}
  \caption{The  variation of the active set size measured by
  $|\calA_s \setminus \calA_{s+1}|$ and  $|\calA_{s+1} \setminus \calA_{s}|$ along the
  continuation path, where $\calA_s$ denotes the active set at the  regularization parameter
   $\lambda_s$.   }
   \label{fig:symdiffas}
\end{figure}

Correspondingly, the variation of the relative $\ell^2$ error with respect to the oracle
solution $x^o$ along the continuation path is given in Fig. \ref{fig:conv-oracle}. For large
regularization parameters $\lambda_s$, the regularized solution is zero, and thus 
the relative error is unit. Then the error 
first increases slightly, before it starts to decrease monotonically.
Upon convergence (i.e., the discrepancy principle is satisfied), the iterate converges to the oracle
solution $x^o$, as indicated by the extremely small error. It is noteworthy that the convergence behavior is almost
identical for both the uncorrelated and correlated sensing matrices, further confirming the advantage
of the $\ell^0(\ell^2)$ approach.
\begin{figure}
\centering
   \begin{tabular}{cc}
   \includegraphics[trim = 0cm 0cm 1cm 0cm, clip=true,width=7cm]{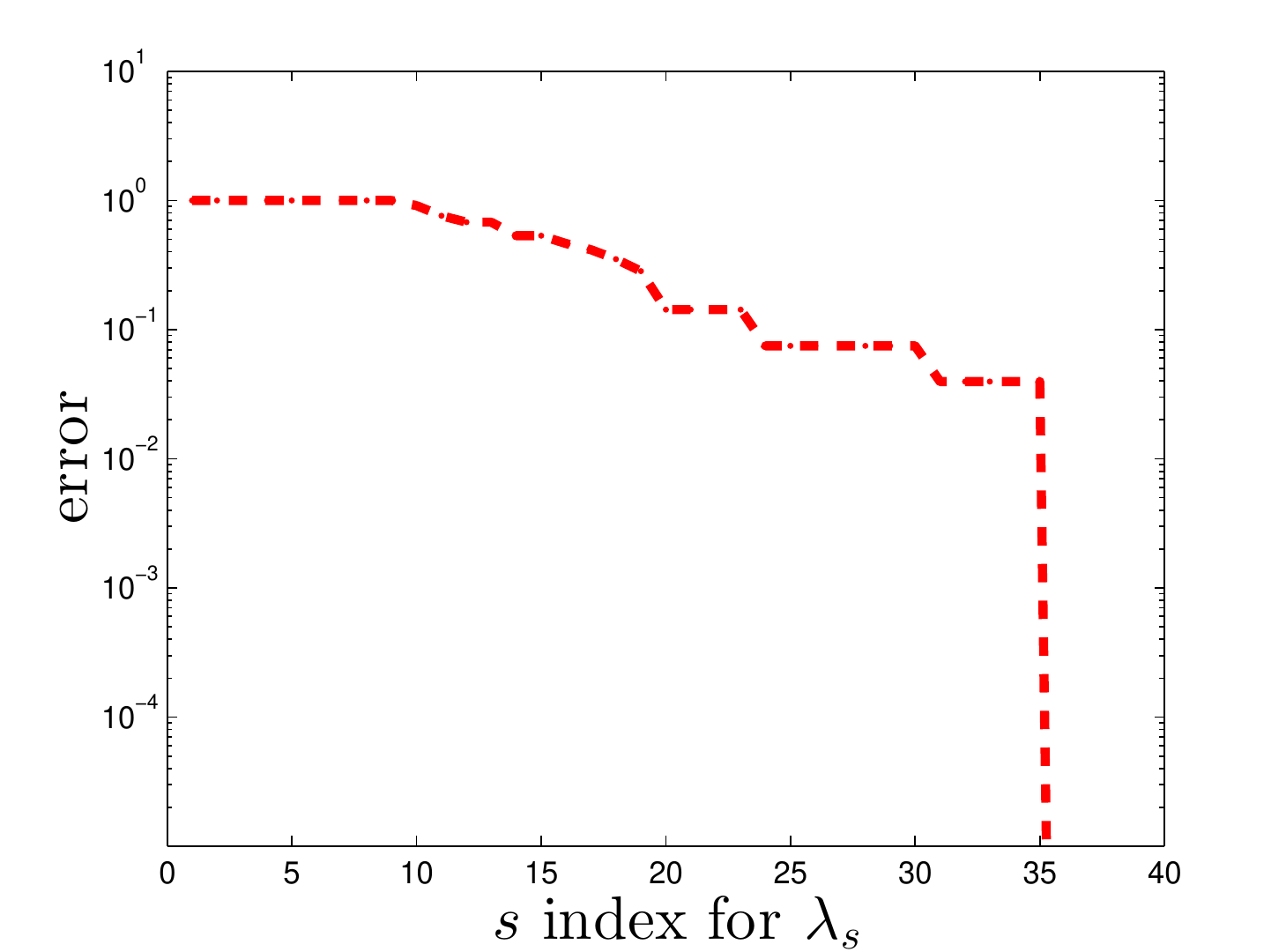} \\
   (a) $(500,10^3,250,50,4,100,0,10^{-3})$ \\
   \includegraphics[trim = 0cm 0cm 1cm 0cm, clip=true,width=7cm]{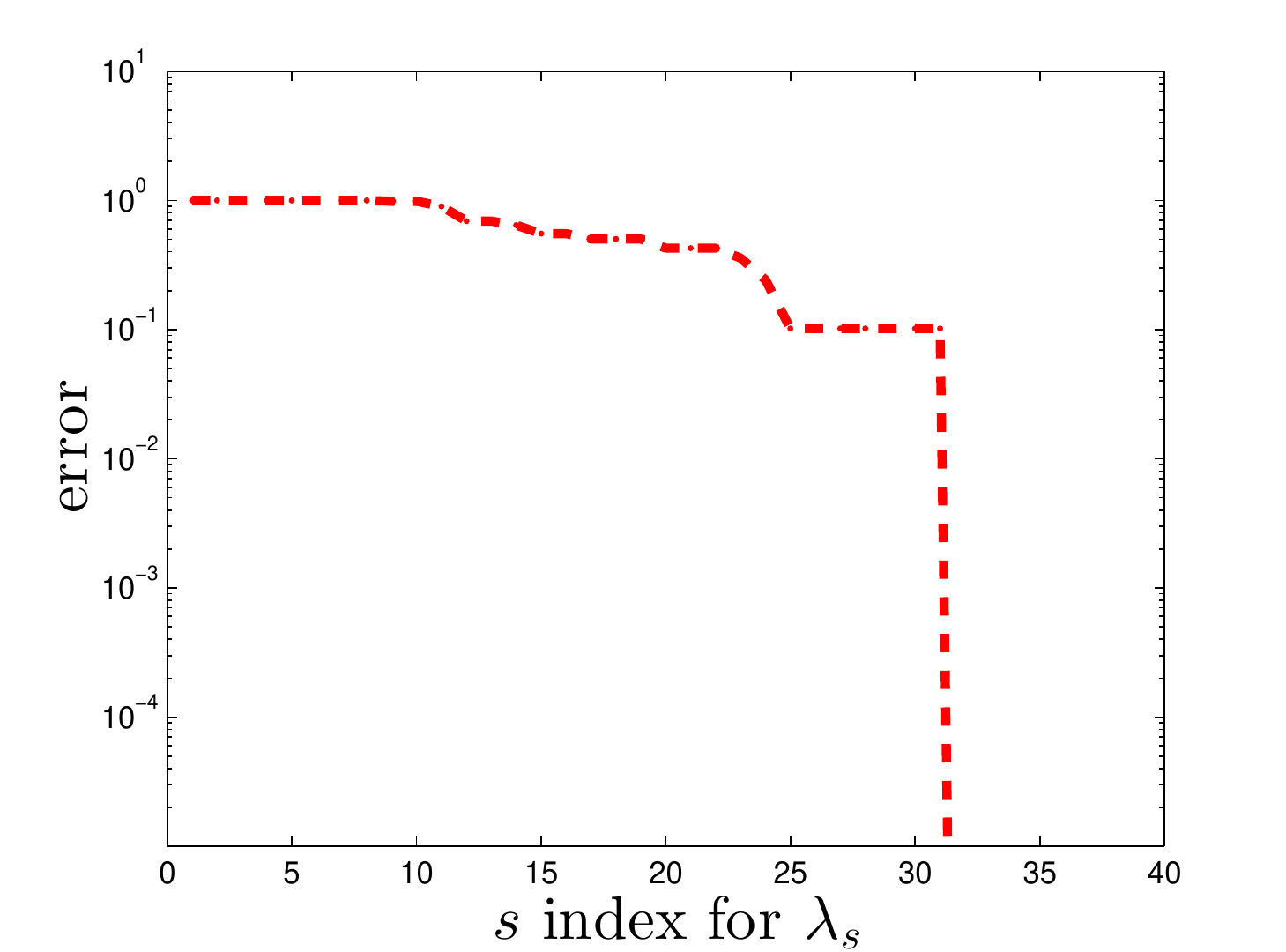}\\
    (b) $(500,10^3,250,50,4,100,1,10^{-3})$
  \end{tabular}
  \caption{The relative $\ell^2$ error of the iterates along the continuation path, for each fixed regularization parameter
   $\lambda_s$, with respect to the oracle solution $x^o$.}\label{fig:conv-oracle}
\end{figure}

\subsection{Multichannel image reconstruction}

In the last set of experiments, we consider recovering 2D images from compressive and noisy measurement.

The first example is taken from \cite{HuangHuangMetaxas:2009}. The target signal is a color image
with three-channels $I = (I_r;I_g;I_b)$, with $I_{c}\in \mathbb{R}^{l^2}, c \in \{r,g,b\}.$ In
the computation, we reorder $I$ into one vector such that the pixels at the same position from the
three channels are grouped together. The observational data $y$ is generated by $y  = \Psi I + \eta$
where $\Psi$ is a random Gaussian matrix (with correlation within  each group) and  $\eta$  is Gaussian noise,
following the procedure outlined in Section \ref{ssec:setup} with the following parameters: $n=1152$, $p=6912$, $N = 2304$, $T=152$,
$s = 3$, $\theta = 10$, $\sigma=\mbox{1e-3}$. The condition number within each group is $O(10^2)$.

The numerical results are presented in Fig. \ref{fig:2d} and Table \ref{tab:2d}, where the PSNR
is defined by
\begin{equation*}
 \mathrm{PSNR}=10\cdot \log\frac{V^2}{MSE},
\end{equation*}
where $V$ and $MSE$ is the maximum absolute value and the mean squared error, respectively, of the
reconstruction, It is observed that GPDASC, GOMP and GCD produce
visually equally appealing results, and they are much better than that of SPGl1. This observation is
also confirmed by the PSNR values in Table \ref{tab:2d}: the PSNR of GPDASC is slightly higher
than that of GOMP and GCD.  The convergence of GPDASC is much faster than GOMP and GCD.
The SPGl1 is the most efficient one, but greatly compromises the reconstruction quality.

\begin{figure}[ht!]
  \centering
  \begin{tabular}{c}
    \includegraphics[trim = 5cm 9.5cm 1cm 8.5cm, clip=true,width=11.5cm]{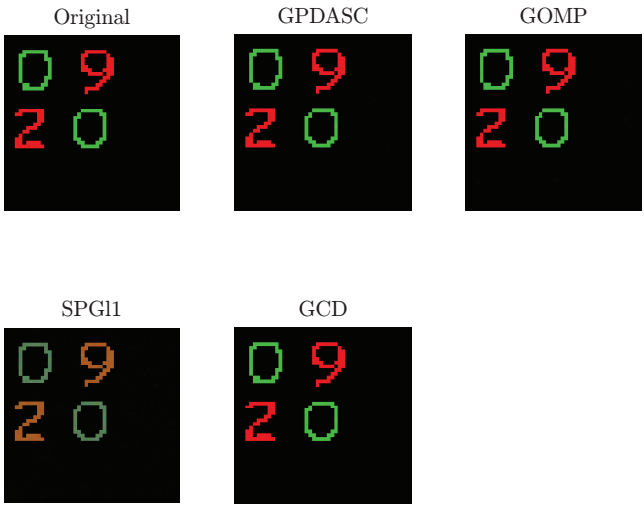}
  \end{tabular}
  \caption{Reconstruction results of the two-dimensional  image.}\label{fig:2d}
\end{figure}

\begin{table}[h]
\centering
 \caption{Numerical results for the two-dimensional image: $n=1152$, $p=6912$, $N = 2304$, $T=152$, $s = 3$, $\theta = 10$, $\sigma=\mbox{1e-3}$.}\label{tab:2d}
 \begin{tabular}{ccccc}
 \hline
    algorithm   &CPU time (s) &PSNR      \\
 \hline
  GPDASC     & 5.70  &48.2    \\
  GOMP       & 10.9  &47.9       \\
  SPGl1      & 2.85  &22.2   \\
  GCD        & 33.9  &48.1     \\
  \hline
  \end{tabular}
\end{table}

Last, we consider multichannel MRI reconstruction. The sampling matrix $\Psi$ is the composition
of a partial FFT with an inverse wavelet transform, with a size $3771\times 12288$, where we have
used 6 levels of Daubechies 1 wavelet. The three channels for each wavelet expansion are organized
into one group, and the underlying image $I = (I_r;I_g;I_b)$ has $724$ nonzero group coefficients
(each of group size $3$) under the wavelet transform. Hence, the data is formed as $y=\Psi c+\eta$,
where $c$ is the target coefficient with a group sparse structure and $\eta$ is the Gaussian noise
with a noise level $\sigma=\mbox{1e-2}$. The recovered image $I$ is then obtained by applying the
inverse wavelet transform to the estimated coefficient $c$. The numerical results are presented in
Fig. \ref{fig:2dwavelet} and Table \ref{tab:2dwavelet}.

\begin{figure}[ht!]
  \centering
  \begin{tabular}{c}
    \includegraphics[trim = 2cm 1cm 1cm 0cm, clip=true,width= 9cm]{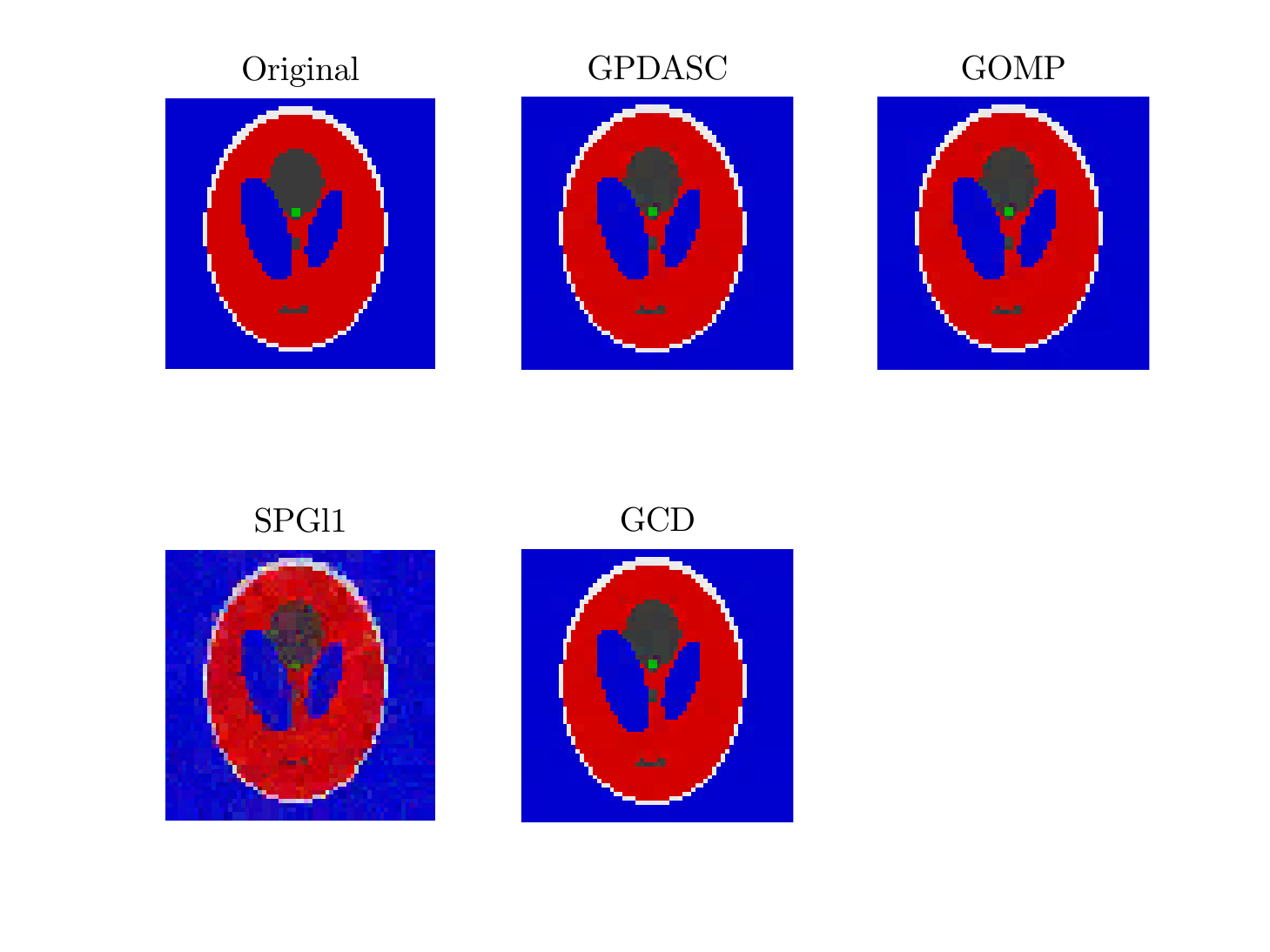}
  \end{tabular}
  \caption{Reconstructions for the 2D MRI phantom image.}\label{fig:2dwavelet}
\end{figure}

\begin{table}[h]
\centering
 \caption{Numerical results for the 2D MRI phantom image: $n=3771$, $p=12288$, $N = 4096$, $T=724$, $s = 3$, $\theta = 0$, $\sigma=\mbox{1e-2}$.}\label{tab:2dwavelet}
 \begin{tabular}{ccccc}
 \hline
    algorithm   &CPU time (s) &PSNR      \\
 \hline
  GPDASC     & 48.5 &38.7   \\
  GOMP       & 203  &37.3       \\
  SPGl1      & 14.3  &20.1   \\
  GCD        & 212  &38.2     \\
  \hline
  \end{tabular}
\end{table}

The observations from the preceding example remain largely valid: the reconstructions
by  GPDASC, GOMP and GCD are close to each other visually and have comparable PSNR
values, and all are much better than that by SPGl1. However, GPDASC
is a few times faster than that by GOMP and GCD.

\section{Conclusions}
In this work we have proposed and analyzed a novel approach for recovering group sparse signals based on the regularized
least-squares problem with an $\ell^0(\ell^2)$ penalty. We provided a complete theoretical analysis on the model,
e.g., existence of global minimizers, invariance property, support recovery,  and
properties of block coordinatewise minimizers. One salient feature of the approach is that it has built-in decorrelation
mechanism, and can handle very strong inner-group correlation. Further, these nice properties can be
numerically realized efficiently by a primal dual active set solver, for which a finite-step global convergence was also
proven. Extensive numerical experiments were presented to illustrate the salient features of the $\ell^0(\ell^2)$
model, and the efficiency and accuracy of the algorithm, and the comparative study with existing approaches
show its competitiveness in terms of support recovery, reconstruction errors and computing time.

There are several avenues deserving further study. First, when the column vectors in each group are ill-posed
in the sense that they are highly correlated / nearly parallel to each other, which are
characteristic of most inverse problems \cite{AlbertiAmmariJin:2016}, the propose $\ell^0(\ell^2)$ model \eqref{eqn:groupl0} may
not be well defined or the involved linear systems in the GPDAS algorithm can be challenging
to solve directly. One possible strategy is to apply an extra regularization. This necessitates a refined theoretical
study. Second, in practice, the true signal may have extra structure within the group,
e.g., smoothness or sparsity. It remains to explore such extra a priori information.

\section*{Acknowledgements}
The authors would like to thank the two referees for their constructive
comments. The research of Y. Jiao is partially supported by
National Science Foundation of  China No. 11501579, B. Jin by
EPSRC grant EP/M025160/1, and X. Lu
by National Science Foundation of  China No. 11471253.

\appendix
\subsection{Proof of Proposition \ref{prop:bmic}}\label{app:bmic}
\begin{proof}
Let $\mathcal{N}_1 = \textrm{span}\{p_1,...,p_{s_1}\}$ and $\mathcal{N}_2 = \textrm{span}
\{q_1,...,q_{s_2}\}$ be two subspaces spanned by two distinct groups, where $p_i$, $q_j$ are
column vectors of unit length. By the definition of the MC $\nu$, $|\langle p_i,q_j\rangle|\leq \nu$ for
any $i=1,\ldots,s_1$ and $j=1,\ldots,s_2$. For any $u\in \mathcal{N}_1$ and $v\in\mathcal{N}_2$, let $u = \sum_{i=1}^{s_1} c_i p_i$
and $v = \sum_{j=1}^{s_2} d_j q_j$. Then with $c = (c_1,...,c_{s_1})$ and $d = (d_1,...,d_{s_2})$,
\begin{equation*}
  \begin{aligned}
    \|u\|^2 & = \sum_{i,j=1}^{s_1} c_ic_j \langle p_i,p_j\rangle \geq \sum_{i=1}^{s_1}c_i^2 - \nu\sum_{i\neq j}|c_i||c_j|\\
    &\geq(1-(s_1-1)\nu)\|c\|^2\geq (1-(s-1)\nu)\|c\|^2,
  \end{aligned}
\end{equation*}
and similarly $\|v\|^2 \geq (1-(s-1)\nu)\|d\|^2$. Hence we have
\begin{eqnarray*}
\frac{|\langle u, v \rangle|}{\|u\| \|v\|} \leq \frac{\nu \sum_{i=1}^{s_1}\sum_{j=1}^{s_2} |c_i d_j|}{ (1 - \nu(s-1))\|c\| \|d\|} \leq \frac{\nu s}{1 - \nu(s-1)},
\end{eqnarray*}
by the inequality
$
   \sum_{i=1}^{s_1}\sum_{j=1}^{s_2}|c_id_j|=\sum_{i=1}^{s_1}|c_i|\sum_{j=1}^{s_2}|d_j|\leq \sqrt{s_1s_2}\|c\|\|d\|\leq s\|c\|\|d\|.
$
\end{proof}

\subsection{Proof of Lemmas \ref{lem:est-G}  and \ref{lem:est-D}}\label{app:G}
\begin{proof}{[of Lemma \ref{lem:est-G}]}
First, recall that for any matrix $A$, $A^tA$ and $A A^t$ have the same
nonzero eigenvalues. Upon letting $A = \bar{\Psi}_{G_i}^{-1}\Psi^t_{G_i}$, we have $AA^t = I$, and
\begin{equation*}
  \|\Psi_{G_i}\bar{\Psi}_{G_i}^{-1} x_{G_i}\|^2 =x^t_{G_i} AA^t x_{G_i} = \|x_{G_i}\|^2,
\end{equation*}
and likewise
\begin{equation*}
   \|\bar{\Psi}_{G_i}^{-1}\Psi^t_{G_i} y\|^2 = y^t A^tA y \leq \lambda_{max}(A^tA) \|y\|^2 = \|y\|^2,
\end{equation*}
giving the first two estimates. If $i=j$, $D_{i,j}$ is an identity matrix, and thus $\|D_{i,j}x_{G_j}\|
= \|x_{G_j}\|$. For $i\neq j$, $U_i = (\bar{\Psi}_{G_i}^{-1}\Psi_{G_i}^t)^t \in \mathbb{R}^{n\times |s_i|}$,
$V_j = (\bar{\Psi}_{G_j}^{-1}\Psi_{G_j}^t)^t \in \mathbb{R}^{n \times |s_j|}$, then
\begin{equation*}
 D_{i,j} = U_{i}^tV_j,\quad U_{i}^tU_i = I,\quad  V_{j}^t V_j = I.
\end{equation*}
Thus by Lemma \ref{equdef}, there holds
\begin{eqnarray*}
\|D_{i,j}x_{G_j}\| =\|  U_{i}^tV_j x_{G_j} \| \leq \|  U_{i}^tV_j\|\| x_{G_j} \| \leq \mu \|x_{G_j}\|,
\end{eqnarray*}
showing the last inequality.
\end{proof}

\begin{proof}{[of Lemma \ref{lem:est-D}]}
Since $D_{i,i} = I$, we have
\begin{equation*}
y = Dx = \left(\begin{array}{c}
x_{G_{i_1}} + \sum_{j\neq i_1} D_{i_1,i_j}x_{G_{i_j}}  \\
\vdots \\
x_{G_{i_M}} + \sum_{j\neq i_M} D_{i_1,i_j}x_{G_{i_j}}
\end{array}\right) = \left(\begin{array}{c}
y_{G_{i_1}} \\
\vdots \\
y_{G_{i_M}}
\end{array}\right).
\end{equation*}
By Lemma \ref{lem:est-G}, $\|D_{k,i_j} x_{G_{i_j}}\| \leq \mu \|x_{G_{i_j}}\|$ for
any $k\neq i_j$. Let $k^*$ be the index such that $\|y\|_{\ell^\infty(\ell^2)} = \|y_{G_{k^*}}\|$.
Then
\begin{equation*}
  \begin{aligned}
 &\|y\|_{\ell^\infty(\ell^2)} = \|y_{G_{k^*}}\| \leq  \|x_{G_{k^*}}\| +  \sum_{i_j\neq k^*} \|D_{k^*,i_j}x_{G_{i_j}}\| \\
   & \leq \|x_{G_{k^*}}\| +  \mu\sum_{i_j\neq k^*} \|x_{G_{i_j}}\| \leq  (1 + (M-1)\mu)\|x\|_{\ell^\infty(\ell^2)}.
  \end{aligned}
\end{equation*}
To show the other inequality, let $j^*$ be the index such that $\|x\|_{\ell^\infty(\ell^2)} = \|x_{G_{j^*}}\|$. Then by Lemma \ref{lem:est-G}, we deduce
\begin{equation*}
  \begin{aligned}
&\|y\|_{\ell^\infty(\ell^2)}  \geq \|y_{G_{j^*}}\| \geq  \|x_{G_{j^*}}\| -  \sum_{i_j\neq j^*} \|D_{j^*,i_j}x_{G_{i_j}}\| \\
& \geq \|x_{G_{j^*}}\| -  \mu\sum_{i_j\neq j^*} \|x_{G_{i_j}}\| \geq  (1 - (M-1)\mu)\|x\|_{\ell^\infty(\ell^2)}.
  \end{aligned}
\end{equation*}
This completes the proof of the lemma.
\end{proof}

\subsection{Proof of Corollary \ref{cor:oracle}}\label{app:oracle-unique}
\begin{proof}
Since $\Psi_{G_i}$ has full column rank, problem \eqref{eqn:oracle} is equivalent to
$
\bar{x}^o|_{\cup_{i\in\calA^\dag}G_i } = \mathop\textrm{argmin} \|\sum_{i\in \mathcal{A}^\dag}\Psi_{G_i}\bar{\Psi}_{G_i}^{-1} \bar{x}_{G_i} -y\|^2,\; \bar{x}^o|_{\cup_{i\in\calI^\dag}G_i } = 0,
$
where $\bar{x}_{G_i} = \bar{\Psi}_{G_i}x_{G_i}$. The normal matrix involved in the least-squares problem
on ${\cup_{i\in\calA^\dag}G_i }$ is exactly the matrix $D$ in Lemma \ref{lem:est-D}, with $\{i_1,...,i_M\}
= \mathcal{A}^\dag$. Then the uniqueness of $x^o$ follows from Lemma \ref{lem:est-D}.
\end{proof}

\subsection{Proof of Theorem \ref{thm:existence}}\label{app:existence}
\begin{proof}
Let $\mathfrak{S}=\{B: B = \cup_{i\in \calI} G_{i}, \ \calI\subseteq \{1,2,...,N\}\} $. Then the set $\mathfrak{S}$
is finite. For any nonempty $B \in\mathfrak{S}$, the problem $\min_{\mathrm{supp}(x)
\subseteq B}\|\Psi x-y\|$ has a minimizer $x^*(B)$. Let $T_B^*=\tfrac{1}{2}\|\Psi x^*(B)-y\|^2 + \lambda
\|x^*(B)\|_{\ell^0(\ell^2)}$, and for $B=\emptyset$, let $T_B^*=\frac{1}{2}\|y\|^2$ and $x^*(B)=0$. Then
we denote $T^*=\min_{B\in\mathfrak{S}}T_B^*$, with the minimizing set $B^*$, and $x^*=x^*(B^*)$.
We claim that $J_\lambda(x^*)\leq J_\lambda(x)$ for all $x\in\mathbb{R}^p$. Given any $x\in\mathbb{R}^p$,
let $B\in\mathfrak{S}$ be the smallest superset of $\mathrm{supp}(x)$. Then $\|x^*(B)\|_{\ell^0(\ell^2)}\leq \|x\|_{\ell^0(\ell^2)}$, and further
by construction $\|\Psi x^*(B)-y\|\leq \|\Psi x-y\|$ and hence $J_\lambda(x)\geq J_\lambda(x^*(B))\geq J_\lambda(x^*)$.
\end{proof}

\subsection{Proof of Theorem \ref{thm:oracle}}\label{app:oracle}
\begin{proof}
Since $x^*$ is a global minimizer of $J_\lambda$, we have
\begin{equation*}
\lambda T + \tfrac{1}{2}\epsilon^2 = J_\lambda(x^\dag) \geq   J_\lambda(x^*) \geq \lambda |\calA|.
\end{equation*}
This and the choice of $\lambda$ imply $|\calA| \leq T$. Since any global minimizer
is also a BCWM, by Theorem \ref{thm:support}(i) below, we deduce
$
  \{i\in \mathcal{A}^\dag: \|\bar{x}^\dag_{G_i}\|\geq  2\sqrt{2\lambda} + 3\epsilon\} \subseteq \mathcal{A}.
$
This gives part (i). Next, for $\lambda \in ({\epsilon^2}/{2},
(\min_{i\in \mathcal{A}^\dag}\{\|\bar{x}^\dag_{G_i}\|\} - 2\epsilon)^2/8)$, there holds $\mathcal{A}^\dag\subseteq
\mathcal{A}$ and hence $\mathcal{A}^\dag = \mathcal{A}$. Hence the only global minimizer is the oracle solution $x^o$.
\end{proof}

\subsection{Proof of Theorem \ref{prop:necopt}}\label{app:bcwm}
\begin{proof}
A BCWM $x^*$ is equivalent to the following:
\begin{equation*}
\begin{aligned}
 x_{G_i}^* \in \mathop\textrm{argmin}\limits_{x_{G_i}\in\mathbb{R}^{s_i}}\tfrac{1}{2}\|\Psi_{G_i}x_{G_i}+ \sum_{j\neq i} \Psi_{G_j}x_{G_j}^*-y\|^2 + \lambda \|x_{G_i}\|_{\ell^0(\ell^2)}\\
\end{aligned}
\end{equation*}
for $i=1,\ldots,N$, is equivalent to
\begin{equation*}
\begin{aligned}
x_{G_i}^* \in \mathop\textrm{argmin}\limits_{x_{G_i}\in\mathbb{R}^{s_i}} & \{\tfrac{1}{2}\|\Psi_{G_i}(x_{G_i} - x_{G_i}^*)\|^2 + \lambda \|x_{G_i}\|_{\ell^0(\ell^2)}\\
  &- \langle x_{G_i} - x_{G_i}^*, \Psi_{G_i}^t(y- \Psi x^*)\rangle \}.
\end{aligned}
\end{equation*}
Using the matrices $\bar\Psi_{G_i} =(\Psi_{G_i}^t\Psi_{G_i})^{1/2}$ and the identities
\begin{equation*}
\left\{\begin{aligned}
 \|\Psi_{G_i}(x_{G_i}-x_{G_i}^*)\| &= \|\bar\Psi_{G_i}(x_{G_i}-x_{G_i}^*)\|,\\
\langle x_{G_i} - x_{G_i}^*, \Psi_{G_i}^t(y- \Psi x^*)\rangle &= \langle \bar\Psi_{G_i}(x_{G_i} - x_{G_i}^*), \bar\Psi_{G_i}^{-1}d_{G_i}\rangle,\\
\|x_{G_i}\|_{\ell^0(\ell^2)} &= \|\bar\Psi_{G_i} x_{G_i}\|_{\ell^0(\ell^2)},
\end{aligned}\right.
\end{equation*}
and recalling $\bar x_{G_i} = \bar\Psi_{G_i}x_{G_i}$, $\bar x_{G_i}^* = \bar\Psi_{G_i}x_{G_i}^*$,
and $\bar d_{G_i}^* = \bar\Psi_{G_i}^{-1} d_{G_i}^*$ etc., we deduce
\begin{equation*}
  \bar{x}_{G_i}^*\in \mathop\textrm{argmin}\limits_{\bar x_{G_i}\in\mathbb{R}^{s_i}}\tfrac{1}{2}\|\bar{x}_{G_i}-(\bar{x}^*_{G_i}+\bar{d}_{G_i}^*)\|^2 + \lambda \|\bar x_{G_i}\|_{\ell^0(\ell^2)}.
\end{equation*}
Using the hard-thresholding operator $H_\lambda$, we obtain \eqref{eqn:opt}.
\end{proof}

\subsection{Proof of Theorem \ref{thm:local}}\label{app:bwcm}
\begin{proof}
It suffices to show $J_\lambda(x^*+h)\geq J_\lambda (x^*)$ for all small $h\in\mathbb{R}^p$. Let
$B=\cup_{i\in \mathcal{A}}G_i$. Then
\begin{equation}\label{eqn:min-chara}
  x_B^*\in\arg\min\tfrac{1}{2}\|\Psi_Bx_B^*-y\|^2.
\end{equation}
Now consider a small perturbation $h\in\mathbb{R}^p$ to $x^*$. If $h_{S\setminus B}=0$,
since $\|x^*+h\|_{\ell^0(\ell^2)}=\|x^*\|_{\ell^0(\ell^2)}$ for small $h$, by \eqref{eqn:min-chara},
the assertion holds. Otherwise, if $h_{S\setminus B}\neq0$, then
\begin{equation}\label{eqn:J-positive}
  \begin{aligned}
  J_\lambda(x^*+h) - J_\lambda(x^*) 
  & \geq \lambda - |(h,d^*)|,
  \end{aligned}
\end{equation}
which is positive for small $h$, since
$\|d^*\|_{\ell^\infty(\ell^2)}\leq\sqrt{2\lambda}$, cf. \eqref{eqn:xtd}. This shows the
first assertion. Now if $\Psi_{B}$
has full column rank, then problem \eqref{eqn:min-chara} is strictly
convex. Hence, for small $h\neq0$ with $h_{S\setminus B}=0$
$\|x^*+h\|_{\ell^0(\ell^2)} = \|x^*\|_{\ell^0(\ell^2)}$ and $\|\Psi (x^*+h)-y\|^2 > \|\Psi x^*-y\|^2$.
This and \eqref{eqn:J-positive} show the second assertion.
\end{proof}

\subsection{Proof of Theorem \ref{thm:support}}\label{app:support}

First, we derive crucial estimates on one-step primal-dual
iteration. Here the energy $E$ associated with an active set $\calA$ is defined by
\begin{equation}\label{eqn:energy1}
  E(\mathcal{A})= \max_{j\in \calA^\dag\setminus\calA} \{\|\bar{x}^\dag_{G_j}\|\}.
\end{equation}
These estimates bound the errors in $\bar x$
on $\mathcal{A}$ by the energy $E$ and the noise level $\epsilon$, and similarly
 $\bar d$ on $\calI$.
\begin{lemma}\label{lem:update-in-A}
Let Assumption \ref{assump:mu} hold, and $\calA$ be a given index set with $|\calA|\leq T$, and $\calI=\calA^c$.
Consider the following one-step primal-dual update (with $B=\cup_{i\in\mathcal{A}}G_i$)
\begin{equation}\label{eqn:update}
  x_{{B}} = \Psi^\dag_{B}y,\quad x_{S\setminus{B}}=0,\quad d= \Psi^t(y-\Psi x),
\end{equation}
where $\Psi_B^\dag = (\Psi_B^t\Psi_B)^{-1}\Psi_B^t$ is the pseudo-inverse of $\Psi_B$.
Then with $\mathcal{P}=\calA\cap\calA^\dag$, $\mathcal{Q}=\calA^\dag\setminus\calA$
and $\mathcal{R}=\calA\setminus\calA^\dag$, $E=E(\mathcal{A})$, for the transformed primal variable $\bar x$,
 there holds
\begin{equation}\label{eqn:estx}
  \|\bar x_{G_i}-\bar x_{G_i}^\dag\| \leq \frac{1 }{1 - |\mathcal{A}|\mu } \left( |\mathcal{Q}|\mu E + \epsilon \right)\quad \forall i\in \mathcal{P}\cup\mathcal{R},
\end{equation}
and for the transformed dual variable $\bar d$, there holds
\begin{align}
   \|\bar{d}_{G_i}\| &\leq C_{|\calA|,\mu}\left(\epsilon + \mu |\mathcal{Q}| E\right) + |\mathcal{Q}|\mu E + \epsilon, i\in\calI\cap \calI^\dag,\nonumber\\
 \|\bar{d}_{G_i}\| &\geq \|\bar{x}^\dag_{G_i}\| - (C_{|\calA|,\mu}\left(\epsilon + \mu |\mathcal{Q}| E\right)\label{eqn:estd}\\
   &\quad + (|\mathcal{Q}|-1)\mu E + \epsilon), i\in\calI\cap \calA^\dag.\nonumber
 \end{align}
with $C_{|\calA|,\mu}=|\mathcal{A}|\mu/(1-\mu(|\mathcal{A}|-1))$.
\end{lemma}
\begin{proof}
First, the least squares update step in \eqref{eqn:update} can be rewritten as
\begin{equation*}
x_{B} = (\Psi_{B}^t\Psi_{B})^{-1}\Psi^t_{B}(\Psi_{B} x^\dag_{B} + \sum_{i\in\mathcal{Q}}\Psi_{G_i} x^\dag_{G_i} + \eta).
\end{equation*}
Hence, there holds
\begin{equation}\label{eqn:primal-diff}
x_{B} - x^\dag_{B} = (\Psi_{B}^t\Psi_{B})^{-1}\Psi^t_{B}(\sum_{i\in\mathcal{Q}}\Psi_{G_i} x^\dag_{G_i} + \eta).
\end{equation}
Let $m=|\mathcal{P}|\leq T$, $\ell=|\mathcal{R}|$, then $k=|\mathcal{Q}|=T-m$. Further, we denote the sets $\mathcal{P}$,
$\mathcal{Q}$ and $\mathcal{R}$ by $\mathcal{P}=\{p_1,\ldots,p_m\}$, $\mathcal{Q}=\{q_1,\ldots,q_k\}$ and $\mathcal{R}=\{r_1,\ldots,r_{\ell}\}$.
Then \eqref{eqn:primal-diff} can be recast blockwise,
using $D_{i,j}=\bar\Psi_{G_i}^{-1} \Psi_{G_i}^t\Psi_{G_j} \bar{\Psi}_{G_j}^{-1}$ etc, cf. \eqref{eqn:notation}, as
\begin{equation*}
\begin{aligned}
e&:=\left[\begin{array}{c} \bar{x}_{G_{p_1}}-\bar{x}_{G_{p_1}}^\dag \\ \vdots \\ \bar{x}_{G_{p_m}}-\bar{x}_{G_{p_m}}^\dag \\
\bar{x}_{G_{r_1}} \\ \vdots\\  \bar{x}_{G_{r_{\ell}}}   \end{array}\right]
=\left[\begin{array}{cc}
D_{\mathcal{P},\mathcal{P}} &D_{\mathcal{P},\mathcal{R}}\\
D_{\mathcal{R},\mathcal{R}} &D_{\mathcal{R},\mathcal{R}}
\end{array}\right]
\end{aligned}
\end{equation*}
\begin{equation*}
\begin{aligned}
&\bullet
\left\{\left[\begin{array}{c}
D_{\mathcal{P},\mathcal{Q}}\\
D_{\mathcal{R},\mathcal{Q}}
\end{array}\right]
\left[\begin{array}{c} \bar{x}^\dag_{G_{q_1}} \\ \vdots \\  \bar{x}^\dag_{G_{q_k}}   \end{array}\right] +
\left[\begin{array}{c}
\bar\Psi_{G_{p_1}}^{-1}\Psi_{G_{p_1}}^t  \\
\vdots \\
\bar\Psi_{G_{p_m}}^{-1}\Psi_{G_{p_m}}^t \\
\bar\Psi_{G_{r_1}}^{-1}\Psi_{G_{r_1}}^t \\
\vdots\\
\bar\Psi_{G_{r_\ell}}^{-1}\Psi_{G_{r_k}}^t
\end{array}\right] \eta \right\},
\end{aligned}
\end{equation*}
where the matrices $D_{\mathcal{P},\mathcal{R}}$ etc. are defined by
\begin{equation*}
  D_{\mathcal{P},\mathcal{R}}=\left[\begin{array}{ccc}
    D_{p_1,r_1} & \cdots & D_{p_1,r_{\ell}} \\
   \vdots & \vdots & \vdots \\
   D_{p_m,r_1} & \cdots & D_{p_m,r_{\ell}}
   \end{array}\right].
\end{equation*}
Next we estimate the two terms in the curly bracket, denoted by $\mathrm{I}$ and $\mathrm{II}$ below.
By Lemma \ref{lem:est-G}, we deduce
\begin{equation}\label{eqn:iteeta}
\|\mathrm{II}\|_{\ell^\infty(\ell^2)} \leq\|\eta\|.
\end{equation}
For the first term $\mathrm{I}$, we denote its rows by $z_i=\sum_{j=1}^kD_{i,q_j}\bar x_{G_j}^\dag$, for any $i\in \mathcal{P}\cup\mathcal{R}$.
Since $(\mathcal{P}\cup \mathcal{R})\cap\mathcal{Q}=\emptyset$, we have for $i\in\mathcal{P}\cup\mathcal{R}$
\begin{eqnarray*}
\|z_i\| = \|D_{i,q_1}\bar{x}^\dag_{G_{q_1}} + \cdots + D_{i,q_k}\bar{x}^\dag_{G_{q_k}}\|
 \leq k\mu\max_{1\leq j \leq k}\{\|\bar{x}^\dag_{G_{q_j}}\|\}.
\end{eqnarray*}
Since the ``energy'' $E= \max_{1\leq j \leq k}
\{\|\bar{x}^\dag_{G_{q_j}}\|\}$,
\begin{equation}\label{eqn:itex}
\|z\|_{\ell^\infty(\ell^2)} \leq |\mathcal{Q}|\mu E.
\end{equation}
By Lemma \ref{lem:est-D}, \eqref{eqn:iteeta} and \eqref{eqn:itex} and the triangle inequality,
\begin{equation}\label{eqn:errx}
\|e\|_{\ell^\infty(\ell^2)} \leq \frac{1}{1 - \mu(|\mathcal{P}|+|\mathcal{R}|-1)}\left(\epsilon + \mu |\mathcal{Q}| E\right).
\end{equation}
Notice that $|\mathcal{P}|+|\mathcal{R}| = |\mathcal{A}|$, we show \eqref{eqn:estx}. Next we turn to the transformed dual variable $\bar{d}$.
By the definition, $d = \Psi^t(y - \Psi x)$, and thus for any $i\in \mathcal{I}$, we have
\begin{equation*}
d_{G_i} = \Psi_{G_i}^t(\sum_{j\in\mathcal{P}\cup\mathcal{R}}\Psi_{G_j}(x_{G_j}-x^\dag_{G_j}) - \sum_{i\in\mathcal{Q}}\Psi_{G_i} x^\dag_{G_i} - \eta),
\end{equation*}
which upon some algebraic manipulations yields
\begin{equation*}
  \bar d_{G_i} = \sum_{j\in\mathcal{P}\cup\mathcal{R}} D_{i,j}(\bar x_{G_j}-\bar x^\dag_{G_j}) - \sum_{j\in\mathcal{Q}} D_{i,j}\bar x_{G_{j}}^\dag - \bar\Psi_{G_i}^{-1}\Psi_{G_i}^t \eta.
\end{equation*}
For any $i\in\mathcal{I}\cap \mathcal{I}^\dag$, by Lemma \ref{lem:est-G} and \eqref{eqn:errx}, we have
\begin{equation*}
  \begin{aligned}
    \|\bar d_{G_i}\| & \leq \|\sum_{j\in \mathcal{P}\cup\mathcal{R}} D_{i,j}(\bar{x}_{G_j}-\bar{x}^\dag_{G_j})\| \\
     & \quad + \|\sum_{j\in\mathcal{Q}}D_{i,j}\bar x_{G_j}^\dag\| + \|\bar\Psi_{G_i}^{-1}\Psi_{G_i}^t\eta\|\\
    & \leq \sum_{j\in\mathcal{P}\cup\mathcal{R}}\mu\|\bar{x}_{G_j}-\bar{x}^\dag_{G_j}\| + \sum_{j\in\mathcal{Q}}\mu\|\bar x_{G_j}^\dag\| + \epsilon\\
    & \leq \frac{(|\mathcal{P}|+|\mathcal{R}|)\mu }{1 - \mu(|\mathcal{P}|+|\mathcal{R}|-1)}\left(\epsilon + \mu |\mathcal{Q}| E\right) + |\mathcal{Q}|\mu E + \epsilon.
  \end{aligned}
\end{equation*}
The estimate for $i\in\calI\cap\calA^\dag=\mathcal{Q}$
follows analogously.
\end{proof}

Now we can present the proof of Theorem \ref{thm:support}.
\begin{proof}
First we derive two preliminary estimates using the notation $\mathcal{P}$, $\mathcal{Q}$ and $\mathcal{R}$ from Lemma \ref{lem:update-in-A}.
Since $|\mathcal{A}|\leq T$ and $|\mathcal{Q}|\leq T$, Lemma \ref{lem:update-in-A} and the triangle inequality yield
\begin{equation}\label{eqn:est-x-new}
\|\bar{x}_{G_i}\|
\leq \frac{1 }{1 - T\mu } \left(T\mu E + \epsilon \right) \quad\forall i\in \calA\cap \calI^\dag.
\end{equation}
Likewise, using the inequality
$
    \frac{|\mathcal{A}|\mu }{1 - \mu(|\mathcal{A}|-1)}\left(\epsilon + \mu |\mathcal{Q}| E\right) + |\mathcal{Q}|\mu E + \epsilon\leq \frac{1}{1-T\mu}(T\mu E+\epsilon),
$
we deduce from Lemma \ref{lem:update-in-A}
\begin{equation}\label{eqn:est-d-new}
\|\bar{d}_{G_i}\| \geq \|\bar{x}^\dag_{G_i}\| - \frac{1}{1-T\mu}(T\mu E + \epsilon) \quad\forall i\in\calI\cap \calA^\dag.
\end{equation}
Now we can proceed to the proof of the theorem. For $\mathcal{Q}=\emptyset$,
$\calA = \calA^\dag$ and assertions (i) and (ii) are trivially true. Otherwise, let
$i^*=\{i\in \mathcal{Q}: \|\bar x_{G_i}^\dag\|=\|\bar x_\mathcal{Q}^\dag\|_{\ell^\infty(\ell^2)}\}$.
Then $E=\|\bar x_{G_{i^*}}^\dag\|$. By \eqref{eqn:est-d-new} and inequality \eqref{eqn:xtd} with $i = i^*$, we have
\begin{equation*}
  \sqrt{2\lambda}\geq \|\bar{d}_{G_{i^*}}\| \geq E - \frac{T\mu}{1- T \mu}E -\frac{\epsilon}{1-T\mu}.
\end{equation*}
Consequently, by Assumption \ref{assump:mu}, we deduce
\begin{equation}\label{eqn:E-upper}
   E \leq \frac{1- T\mu}{1- 2T\mu}\sqrt{2\lambda} + \frac{1}{1-2T\mu}\epsilon < 2\sqrt{2\lambda} + 3\epsilon,
\end{equation}
i.e., assertion (i) holds. Next we show assertion (ii) by contradiction. If $\calA\not\subseteq
\calA^\dag$, we can choose $j\in\calA \backslash \calA^\dag$, and apply \eqref{eqn:est-x-new} and
\eqref{eqn:est-d-new}, together with \eqref{eqn:xtd}, to obtain
\begin{equation*}
\frac{1}{1-T\mu}(T\mu E + \epsilon) \geq \|\bar{x}_{G_j}\| \geq \|\bar{d}_{G_{i^*}}\| \geq \frac{1 -2T\mu}{1- T \mu}E -\frac{\epsilon}{1-T\mu},
\end{equation*}
which contradicts \eqref{assump:noise}, thereby showing assertion (ii). Last, we show assertion (iii).
Assume that $\calA\not\subseteq \calA^\dag$. Then \eqref{eqn:E-upper} holds. Meanwhile,
since $\mathcal{A}\cap \mathcal{I}^\dag\neq\emptyset$, using \eqref{eqn:est-x-new} and
\eqref{eqn:est-d-new} (by choosing $\bar{x}_{G_i}$ by $i\in\mathcal{A}\cap \mathcal{I}^\dag$ and
$\bar{d}_{G_{i^*}}$) and inequality \eqref{eqn:xtd}, we have
\begin{equation*}
E - \frac{1 }{1 - \mu T} \left( T\mu E + \epsilon \right) \leq \sqrt{2\lambda} \leq  \frac{1 }{1 - T\mu } \left(T\mu E + \epsilon \right) .
\end{equation*}
Under Assumption \ref{assump:mu}, simple computation gives $E\geq 2\sqrt{2\lambda} - 3\epsilon$ and $E\leq 2\sqrt{2\lambda} + 3\epsilon$.
This contradicts with the assumption in (iii), and thus the inclusion $\calA\subseteq\calA^\dag$ follows.
\end{proof}

\subsection{Proof of Proposition \ref{prop:lam0}}\label{app:lam0}

\begin{proof}
Recall the identity
$
  J_\lambda(x) = \tfrac{1}{2}\|\Psi x-y\|^2 + \lambda \|x\|_{\ell^0(\ell^2)} = J_\lambda(0) + R(x),
$
with $R(x)=\tfrac{1}{2}\|\Psi x\|^2 -\langle \Psi x,y\rangle + \lambda \|x\|_{\ell^0(\ell^2)}$.
Also for any $x\neq 0$, $\|x\|_{\ell^0(\ell^2)}\geq1$. Hence, for any $x\in B_r(0)\setminus\{0\}$, where
$B_r(0)$ denotes a ball centered at the origin with a radius $r=\lambda/(\|\Psi^ty\|+1)$, there holds
$R(x) \geq -\|x\|\|\Psi^ty\| + \lambda >0.$ This shows the first assertion.
For $\lambda>\lambda_0$, for any nonzero $x$, we have $\|x\|_{\ell^0(\ell^2)}\geq 1$, and thus
$ J_\lambda(x) = \tfrac{1}{2}\|\Psi x-y\|^2 + \lambda \|x\|_{\ell^0(\ell^2)} \geq \lambda > \tfrac{1}{2}\|y\|^2=J_\lambda(0),$
i.e., $x^*=0$ is the only global minimizer.
\end{proof}

\subsection{Proof of Theorem \ref{thm:main}}\label{app:07}

\begin{proof}
The lengthy proof is divided into four steps.

\noindent \textbf{Step 1}.  First we give the proper choice of the decreasing factor $\rho$. By \eqref{assump:noise}, we have
\begin{equation*}
0< \frac{1-\mu T}{1- 2\mu T -t} < \frac{1 - \mu T}{\mu T  + t}.
\end{equation*}
Then for any $s_1\in((1-\mu T)/(1-2\mu T -t), (1-\mu T)/(\mu T + t))$, letting $s_2=\frac{\mu T +t}{1-\mu T}s_1+1$, we deduce
$({1-\mu T})/({1-2\mu T -t}) < s_2< s_1<({1-\mu T})/(\mu T +t).$
Combining with the monotonicity of the function $f(s_1)=s_2/s_1$ over the interval $((1-\mu T)/(1-2\mu T -t),
(1-\mu T)/(\mu T + t))$, it implies that for any $\rho\in((2\mu T +2 t)^2/(1-\mu T)^2,1)$, we can find such $s_1$ with $s_2/s_1=\sqrt{\rho}$.
Next we will choose $\rho\in((2\mu T +2 t)^2/(1-\mu T)^2,1)$.

\noindent \textbf{Step 2}. Next we show an important monotonicity relation:
\begin{equation}\label{eqn:key}
\Gamma_{s_1^2 \lambda} \subseteq \calA_k\subseteq \calA^\dag \Rightarrow \Gamma_{s_2^2 \lambda} \subseteq \calA_{k+1}\subseteq \calA^\dag.
\end{equation}
For short, we denote by $\calA=\calA_k$, $\calI=\calI_k$, and $\mathcal{Q}=\calA^\dag\setminus\calA$. By the assumption $\calA\subseteq\calA^\dag$,
we have $\mathcal{R}=\emptyset$ in Lemma \ref{lem:update-in-A}. Then it follows from
\eqref{eqn:errx} in the proof of Lemma \ref{lem:update-in-A} that the
updates $\bar {x}^{k+1}$ and $\bar d^{k+1}$ satisfy
\begin{eqnarray}
     && \|\bar{x}^{k+1}_{G_i}\| \geq \|\bar{x}^\dag_{G_i}\| - \frac{1 }{1 - \mu T} \left( T\mu E_k + \epsilon \right)\;\; \forall i\in\calA,\label{eqn:update-xinA}\\
     && \|\bar{d}^{k+1}_{G_i}\| \leq \frac{1 }{1 - \mu T} \left( T\mu E_k + \epsilon \right)\;\; \forall i\in  \mathcal{I}^\dag,\label{eqn:update-dinI}\\
     && \|\bar{d}^{k+1}_{G_i}\| \geq \|\bar{x}^\dag_{G_i}\| - \frac{1 }{1 - \mu T} \left( T\mu E_k + \epsilon \right)\;\; \forall  i\in\mathcal{Q}.\label{eqn:update-dinQ}
\end{eqnarray}
By the assumption $\Gamma_{s_1^2 \lambda} \subseteq \mathcal{A}_k$, we deduce $E_k < s_1 \sqrt{2\lambda}$; and
by assumption \eqref{assump:noise}, $\epsilon < t\min_{i\in \mathcal{A}^\dag}\{\|\bar{x}^\dag_{G_i}\|\}
\leq tE_k\leq t s_1\sqrt{2\lambda}$. Hence, using \eqref{eqn:update-dinI}, we deduce for any $i\in \mathcal{I}^\dag$
\begin{equation*}
  \|\bar{d}^{k+1}_{G_i}\| \leq \frac{1 }{1 - \mu T} ( T\mu +t)E_k \leq \frac{\mu T +t }{1-\mu T} s_1\sqrt{2\lambda}<\sqrt{2\lambda},
\end{equation*}
where the last inequality follows from the choice of $s_1$. This and
the relation \eqref{eqn:xtd} imply that $i\in \calI_{k+1}$, and thus $\calA_{k+1}\subseteq \calA^\dag$.
Meanwhile, by \eqref{eqn:update-dinQ}, for any $i\in\calI\cap \Gamma_{s_2^2\lambda}$, we have
\begin{equation*}
  \begin{aligned}
   \|\bar{d}^{k+1}_{G_i}\| &\geq s_2 \sqrt{2\lambda} - \frac{1}{1-\mu T}(\mu T +t)s_1\sqrt{2\lambda}\\
   &\geq (s_2-\frac{\mu T +t}{1-\mu T-t}s_1)\sqrt{2\lambda}> \sqrt{2\lambda}.
  \end{aligned}
\end{equation*}
which by the relation \eqref{eqn:xtd} yields $i\in \calA_{k+1}$. It remains to show  $\calA\cap
\Gamma_{s_2^2\lambda}\subseteq \calA_{k+1}$. Clearly, if $\calA=\emptyset$, the assertion is
true. Otherwise, for any $i\in\calA\cap\Gamma_{s_2^2\lambda}$, by \eqref{eqn:update-xinA}, there holds
\begin{equation*}
  \begin{aligned}
    \|\bar x_{G_i}\| & \geq \|\bar x_{G_i}^\dag\| - \frac{|\mathcal{Q}|\mu+t}{1-(T-1)\mu}\|x_\mathcal{Q}^\dag\|_{\ell^\infty(\ell^2)}\\
    &  > s_2\sqrt{2\lambda} - \frac{(T-1)\mu+t}{1-T\mu}s_1\sqrt{2\lambda}\geq \sqrt{2\lambda}.
  \end{aligned}
\end{equation*}
Like before, this and \eqref{eqn:xtd} also imply $i\in \calA_{k+1}$. Hence the inclusion
$\Gamma_{s_2^2\lambda} \subseteq \calA_{k+1}$ holds.

\noindent \textbf{Step 3}. Now we prove that the oracle solution $x^o$ is achieved along the continuation
path, i.e., $\calA(\lambda_s)=\calA^\dag$ for some $\lambda_s$. For each $\lambda_s$-problem
$J_{\lambda_s}$, we denote by $\calA_{s,0}$ and $\calA_{s,\diamond}$ the active set for the initial guess
and the last inner step (i.e., $\calA(\lambda_s)$ in Algorithm \ref{alg:gpdasc}) of the $s$th iterate of
the outer loop, respectively. Since $s_1>s_2$, the inclusion $\Gamma_{s_1^2\lambda_s} \subseteq\Gamma_{s_2^2
\lambda_s}$ holds. Next we claim that the following inclusion by mathematical induction
\begin{equation*}
  \Gamma_{s_1^2\lambda_s}\subseteq \mathcal{A}(\lambda_s) \subseteq \mathcal{A}^\dag
\end{equation*}
holds for the sequence active sets $A(\lambda_s)$ from Algorithm \ref{alg:gpdasc}.
 From \eqref{eqn:key}, for any index $s$ before the stopping criterion
at step 13 of Algorithm \ref{alg:gpdasc} is reached, there hold
\begin{equation}\label{eqn:inclusion}
  \Gamma_{s_1^2\lambda_s}\subseteq \calA_{s,0} \quad\mbox{and}\quad \Gamma_{s_2^2\lambda_s}\subseteq \calA_{s,\diamond}.
\end{equation}
Note that for $s=0$, by the choice of $\lambda_0$, $\Gamma_{s_1^2\lambda_0}=\Gamma_{s_2^2\lambda_0}=\emptyset$,
and thus \eqref{eqn:inclusion} holds. Now for $s>0$, it follows by mathematical induction and
the relation $\calA_{s,\diamond}=\calA_{s+1,0}$. By \eqref{eqn:inclusion}, during the iteration,
the active set $\calA_{s,\diamond}$ always lies in $\calA^\dag$. This shows the desired claim.
For large $s$, we have $\Gamma_{s_1^2\lambda_s} = \calA^\dag$, and hence $\calA(\lambda_s) = \calA^\dag$, and
accordingly $x(\lambda_s)$ is the oracle solution $x^o$.

\noindent \textbf{Step 4}. Last, at this step we show that if $\calA(\lambda_s)\subsetneqq
\calA^\dag$, then the stopping criterion at step 13 of Algorithm \ref{alg:gpdasc} cannot be satisfied. Let
$\mathcal{P} = \calA(\lambda_s) \subsetneqq \calA^\dag$ and $\mathcal{Q} = \calA^\dag \backslash \calA$, and
denote by $i^* = \textrm{arg}\max_{i\in\mathcal{Q}}\{\|\bar{x}^\dag_{G_i}\|\}$ and $E = \|\bar{x}^\dag_{G_{i^*}}\|$.
Then with the notation $\bar{\Psi}_{G_i}$ and $D_{i,j}$ etc. from \eqref{eqn:notation}, we deduce
\begin{equation*}
  \begin{aligned}
 &\|\Psi x - y\|^2 = \|\sum_{i\in\mathcal{P}} \Psi_{G_i} (x_{G_i} - x_{G_i}^\dag)  - \sum_{j\in\mathcal{Q}} \Psi_{G_j} x^\dag_{G_j} - \eta \|^2 \\
& \geq \|\Psi_{G_{i^*}}x^\dag_{G_{i^*}}\|^2+2\sum_{j\in\mathcal{Q}\setminus\{i^*\}}\langle\Psi_{G_j} x^\dag_{G_j},\Psi_{G_{i^*}} x^\dag_{G_{i^*}} \rangle \\
  &\quad - 2\sum_{i\in \mathcal{P}}\langle \Psi_{G_i} (x_{G_i} - x_{G_i}^\dag),\Psi_{G_{i^*}} x^\dag_{G_{i^*}}\rangle  + 2\langle\eta, \Psi_{G_{i^*}}{x}^\dag_{G_{i^*}}\rangle.
\end{aligned}
\end{equation*}
Now recall the elementary identities
$\|\Psi_{G_{i^*}}x_{G_{i^*}}\| = \|\bar x_{G_{i^*}}\|$ and $\langle\Psi_{G_j} x^\dag_{G_j} , \Psi_{G_{i^*}} x^\dag_{G_{i^*}}\rangle=\langle D_{i^*,i} \bar x^\dag_{G_i},\bar x^\dag_{G_{i^*}}\rangle$
and then appealing to Lemma \ref{lem:est-D}, we arrive at
\begin{equation*}
  \begin{aligned}
   & \|\Psi x - y\|^2  \geq \|\bar x^\dag_{G_{i^*}}\|^2+2\sum_{j\in\mathcal{Q}\setminus\{i^*\}}\langle D_{i^*,j} \bar x^\dag_{G_j},\bar x^\dag_{G_{i^*}} \rangle\\
    &\quad - 2\sum_{i\in \mathcal{P}}\langle D_{i^*,i}(\bar x_{G_i} - \bar x_{G_i}^\dag), \bar x^\dag_{G_{i^*}}\rangle  + 2\langle\eta, \Psi_{G_{i^*}}{x}^\dag_{G_{i^*}}\rangle \\
    & \geq E^2 - 2\mu(|\mathcal{Q}|-1) E^2 - 2\mu |\mathcal{P}| E \max_{i\in\mathcal{P}} \|\bar{x}_{G_i} - \bar{x}^\dag_{G_i}\| - 2\epsilon E.
  \end{aligned}
\end{equation*}
By repeating the proof of Lemma \ref{lem:update-in-A}, we deduce
\begin{equation*}
\max_{i\in\mathcal{P}} \|\bar{x}_{G_i} - \bar{x}^\dag_{G_i}\| \leq \frac{\epsilon + \mu T E}{1 - \mu T}.
\end{equation*}
By assumption \eqref{assump:noise}, $\epsilon \leq t E$, it suffices to show
\begin{equation}\label{eqn:energy-iter}
 E^2 - 2\mu (|\mathcal{Q}|-1) E^2 - 2\mu(T - |\mathcal{Q}|)\frac{t + \mu T }{1 - \mu T}E^2 - 2t E^2 > t^2 E^2,
 \end{equation}
which implies that the stopping criterion \eqref{eqn:discprin} at step 13 of Algorithm
\ref{alg:gpdasc} cannot be satisfied. The left hand side of \eqref{eqn:energy-iter}
is a function monotonically decreasing with respect to the length $|\mathcal{Q}|$, and
when $|\mathcal{Q}| = T$, we have $1 - \mu(T-1) - 2t > t > t^2$,  which completes the proof.
\end{proof}

\bibliographystyle{IEEEtran}
\bibliography{IEEEabrv,group_v4}
\end{document}